\documentclass[journal, comsoc, 10pt]{IEEEtran}

\usepackage{soul}

\usepackage[T1]{fontenc}% optional T1 font encoding

\usepackage{amssymb,amsmath,amsfonts,amsthm}
\usepackage{mathrsfs}
\usepackage{cite}
\usepackage{array}
\usepackage{tabularx}
\usepackage[table]{xcolor}
\usepackage{multicol, multirow}
\usepackage{color}
\usepackage{mathtools}
\usepackage{subcaption}
\usepackage{mathtools}
\usepackage{tikz,pgf}
\usetikzlibrary{calc,positioning,mindmap,trees,decorations.pathreplacing}
\usepackage{graphicx}
\usepackage{bbm}
\usepackage[utf8]{inputenc}
\usepackage{pgfplots}
\usepackage{algorithm,algcompatible}
\usepackage{circledsteps}

\newtheorem{propi}{Proposition}
\newtheorem{remark}{Remark}

\newtheorem{lemma}{Lemma}
\newtheorem{assumption}{Assumption}
\newtheorem{theorem}{Theorem}

\newcommand{\argmin}{\mathop{\rm arg~min}\limits}

\algblockdefx{FORALLP}{ENDFAP}[1]%
{\textbf{for all }#1 \textbf{do in parallel}}%
{\textbf{end for}}

\usepackage[colorinlistoftodos,prependcaption,backgroundcolor=black!5!white,bordercolor=red]{todonotes}

\IEEEoverridecommandlockouts
\ifCLASSINFOpdf
\else
\fi
\usepackage[cmintegrals]{newtxmath}
\begin{document}
	%\title{Federated Learning in the Downlink with Block-based Differential Coding}
	\title{Mixed-Timescale Differential Coding for Downlink Model Broadcast in Wireless Federated Learning}
	
	\author{Chung-Hsuan Hu, Zheng Chen, and Erik G. Larsson
	\thanks{The authors are with the Department of Electrical Engineering (ISY), Link\"{o}ping University, Link\"{o}ping, SE-58183 Sweden. E-mail:\{chung-hsuan.hu, zheng.chen, erik.g.larsson\}@liu.se. This work was supported in part by Zenith, ELLIIT, the Swedish Research Council (VR), and the Knut and Alice Wallenberg (KAW) Foundation.\\
    A preliminary version of this paper was presented at the 2024 Asilomar conference on Signals, Systems, and Computers \cite{ch2024mix}.}
    }
	
	\maketitle
	
\begin{abstract}
	In standard federated learning systems, the parameter server broadcasts the global model to the participating devices in every iteration. Motivated by the temporal correlation between consecutive global models, differential coding can be applied to global model dissemination to reduce the information magnitude, thereby enabling communication with fewer quantization bits. However, due to wireless link failures, devices may occasionally miss differential updates and consequently fail to reconstruct the global model. As a result, they either continue local training based on an outdated model or remain idle until the next full-model broadcast becomes available. To address this challenge, we propose a mixed-timescale differential coding (MTDC) scheme that performs differential coding at two different levels by adjusting the reference model. With MTDC, a device can reconstruct the latest global model between two full-model broadcasts even if it misses a differential update. We provide a convergence analysis that motivates the design of an age-aware variant of MTDC, along with a device scheduling policy to further improve communication efficiency. Simulation results demonstrate that the proposed MTDC schemes achieve superior learning performance compared to baseline methods under similar communication resource budgets in the presence of downlink transmission failures.
\end{abstract}


\section{Introduction}
Federated learning (FL) is a popular distributed machine learning framework in which a set of devices collaboratively train a model through iterative local training and centralized aggregation at a server \cite{mcmahan2017communication}. The iterative exchange of model information between the server and devices incurs substantial communication overhead, making communication efficiency a key challenge in FL systems. For FL over wireless networks, a wide range of communication-efficient transmission strategies have been investigated in the literature \cite{gafni2022fed,li2020fed}. In particular, device scheduling and resource allocation can be optimized using data-importance or learning-aware metrics \cite{luo2022tack,amiriconvergence,dogan2025opt,yang2024DetFed}, as well as by accounting for communication uncertainties such as time-varying channel conditions and interference \cite{salehi2021federated,importance-aware,Ozfatura2021fastF,chen2022federated,scheduling-fl-latency,chen2020ajoint}. In addition, data compression techniques, including sparsification \cite{oh2023comm,alistarh2018convergence} and quantization \cite{alistarh2017qsgd,jhu2021ada}, have been widely adopted to further improve communication efficiency.

% to further reduce the size of transmitted messages. 

%Recent studies have explored  differential or predictive coding  \cite{yue2022com, adrian2024Temp,Adi2021comp,song2024resfed, zheng2021design} that leverage the temporal correlation of local models or gradients for efficient compression. 
%Reference \cite{zheng2021design} emphasizes the importance of transmitting a weight differential in achieving efficient data compression. Based on the same premises, \cite{song2024resfed,adrian2024Temp,yue2022com} investigate predictive coding schemes for the differential vector, with reduced communication load accomplished by adaptive quantization or event-triggered information updating.
%Furthermore, some approaches advocate using analog superposition to alleviate heavy uplink (UL) traffic.
%\zheng{Second paragraph, not many papers consider reducing downlink transmission. Downlink broadcast frequency can be reduced by applying different coding. Explain briefly how it works, background and basic ideas (temporal correlation makes it useful).}

Most prior work on communication-efficient FL has focused on uplink (UL) transmission from devices to the server, while the downlink (DL) transmission from the server to the devices has received   less attention. On the DL, the server can, in principle, broadcast the global model so that all devices receive it simultaneously using the same time–frequency resources. In contrast, on the UL, communication resources need to be divided among transmitting devices such that the server can receive multiple local model updates without interference.
DL communication efficiency has been studied in \cite{amiri2022Con}, which focuses on the comparison between analog (based on over-the-air computation \cite{ota-chen}) and digital transmission designs.
Joint UL and DL communication-efficient designs have been investigated in terms of data compression \cite{cui2021slash,hu2023flex}, resource allocation \cite{mu2022comm}, and analog transmission designs \cite{zhang2023ota}. 
%Reference \cite{cui2021slash} aims for a unified compression scheme applicable to UL and DL transmissions, while \cite{hu2023flex} tackles the joint optimization of communication resources and compression levels for both UL and DL transmissions.
%\zheng{Do we need so many references on communication-efficient FL? No need to write their explicit contributions if they are not closely related.}
%Reference \cite{zhang2023ota} minimizes the distortion of transmitted information by designing the device equalization coefficients and the server-side forwarding matrix.
%Considering limited communication resources, \cite{mu2022comm} proposes an optimization framework for UL/DL transmit power. 
%\zheng{Add more references on differential or predictive coding for federated learning. These papers most likely will have our potential reviewers.}
%Reference 

%Our focus herein is on improviFL over wireless networks with unreliable communication links.
%However, for such DL broadcast transmission to work, it has to be heavily protected against fading, noise, and interference. 
This paper focuses on the DL transmission of global models in wireless FL with unreliable communication links, and specifically with the use of differential coding (DiC)  for the global model compression.
DiC is widely used for image and video coding  as a lossy source coding scheme that exploits temporal and spatial redundancy \cite{sayood2000introduction}. The key idea is to encode the differences between adjacent pixels, consecutive frames, or prediction residuals rather than the absolute values. By doing so, the entropy of the signal can be greatly reduced, enabling more efficient compression and lower bit rates.

In FL, existing studies have illustrated the evidence of temporal correlation between learning models in consecutive iterations and proposed various DiC-based techniques for gradient or model update compression \cite{yue2022com, adrian2024Temp,Adi2021comp,song2024resfed, zheng2021design}. Nevertheless, all these studies focus on the uplink transmission of model updates from local devices to the server. The potential advantages of using DiC for DL communication cost reduction have not been   explored. 
One drawback of the standard DiC framework is that it is prone to decoding failures. The reconstruction of the original frame critically depends on whether the differential update (or residual) is correctly received.
In the context of FL, missing one (differentially encoded) model broadcast hinders the reconstruction of the current and subsequent global   models.
Existing works mentioned above have not  addressed this issue, which is the main motivation behind our work.
%If a device misses an update, it has to conduct model training based on an outdated model, which adversely affects the learning process. 

\subsection{Contributions}
We propose a novel mixed-timescale differential coding (MTDC) scheme to robustify DiC-based compression for DL global model transmission in wireless FL.
%\zheng{The meaning of "the full model" and "the model" and their differences are not clear.}
MTDC is a hierarchical differential coding scheme operating at three timescales 
(with different quantization levels): (1) the full model is broadcast on a slow timescale (i.e., infrequently), without DiC and at high resolution; (2) a first-level DiC-coded model update is broadcast with fewer quantization bits, on an intermediate timescale; and (3) a second-level DiC-coded model update is broadcast with even fewer quantization bits, on a fast timescale.
The difference between the two levels of differential updates lies in their choice of reference frame: the first-level update uses an older version of the global model -- which is more likely to have been correctly reconstructed at the local devices -- as the reference for computing the residual (differential update). 
In this way, the first-level differential update (at the intermediate timescale) serves as a fallback mechanism for the second-level differential update (at the fast timescale), 
and the full-model broadcast serves as a fallback mechanism for the first-level differential update, in case of transmission failures.
%More specifically, an agent that fails to decode a fast-timescale differential update will have to wait for the next intermediate-timescale differential broadcast. Similarly, an agent that fails to decode an intermediate-timescale differential broadcast will have to wait for the next full-model broadcast.

%% -- redo --
%Conceptually, MTDC introduces a fallback mechanism to ameliorate overall model staleness across devices and therefore mitigates the effect of outdated updates in the learning process. 
%Recalling that we propose MTDC to tackle the issue of limited communication resources. For the same level of magnitude resolution, transmission of a full model requires more resources compared to that of a differential model. Hence, the broadcast frequency of a full model should be kept low to save resources, but at the same time, not to sacrifice the learning performance. For example, in milder decoding-failure scenarios, we can schedule fewer full-model broadcasts, and vice versa.

In addition to the MTDC scheme, we propose an age-aware extension, termed A-MTDC, which adaptively selects the type of DL broadcast (full model or differential update) based on the model staleness level at local devices and the statistics of decoding failures. In the presence of DL decoding failures, we establish the convergence of Federated Averaging under the proposed MTDC scheme (see Section~\ref{sec:conv_study} for  specifics). This theoretical result further motivates the design of A-MTDC and the associated age-aware scheduling policy.
%This analysis also  motivates the design of A-MTDC and the age-aware scheduling policy.

\section{System Model}	
We consider a wireless FL system with a parameter server and a set of devices $\mathcal{K}\triangleq\{1,...,K\}$ participating in
the training of a shared learning model parameterized by $\boldsymbol{\theta}\in\mathbb{R}^d$. 
The goal is to minimize a global loss function $F(\boldsymbol{\theta})\triangleq\sum_{k\in\mathcal{K}}w_k F_k(\boldsymbol{\theta})$, which is a weighted average of local loss functions $F_k(\boldsymbol{\theta}), \forall k\in\mathcal{K}$, evaluated over the local datasets $\mathcal{S}_k,k\in\mathcal{K}$. 
A widely used algorithm for this is %To achieve this goal, several algorithms have been proposed, among which 
Federated Averaging (FedAvg) \cite{mcmahan2017communication}.
% is the most widely used. It 
It operates by iteratively combining local training on devices with centralized model aggregation.
%\zheng{Briefly mention about local datasets so that the description is self contained.}

\subsection{Federated Averaging (FedAvg)} 
\label{sec:flbaseline}
At the $t$-th global iteration ($t=1,...$):
\begin{enumerate}
	\item The server broadcasts the global model $\boldsymbol{\theta}(t)$ to $\mathcal{K}$ and schedules a device subset $\Pi(t)\subseteq\mathcal{K}$ for model training.  
	\item Each device $k\in\Pi(t)$ updates the model with $E$-step mini-batch stochastic gradient descent:
    \begin{equation}
        \boldsymbol{\theta}_k(t,\iota+1)=\boldsymbol{\theta}_k(t,\iota)-\eta\nabla F_k(\boldsymbol{\theta}_k(t,\iota);\mathcal{B}_k(t,\iota)),
    \end{equation}
    $\iota=0,...,E-1,$ where $\boldsymbol{\theta}_k(t,0)=\boldsymbol{\theta}(t)$, $\mathcal{B}_k(t,\iota)\subseteq \mathcal{S}_k$, and $\eta$ is the learning rate.
    %\zheng{Each device doesn't exactly minimize $F_k$. The above sentence implies that the updated local model will be the minimizer of the local loss function.}
    The corresponding model update
    \begin{equation}
        \label{eq:triThetak}\triangle\boldsymbol{\theta}_k(t)=\boldsymbol{\theta}_k(t,E)-\boldsymbol{\theta}_k(t,0)
    \end{equation} is transmitted to the server.
    %\zheng{Give proper definition of $\triangle\boldsymbol{\theta}_k(t)$.}
    
	\item The server aggregates the received gradient updates, and updates the global model according to
	\begin{equation}
		\boldsymbol{\theta}(t+1)=\boldsymbol{\theta}(t)+\sum_{ k\in\Pi(t)}w_k\triangle\boldsymbol{\theta}_k(t),
		\label{eq:syncFlAggregation}	
	\end{equation}
	where $\sum_{k\in\Pi(t)}w_k=1$.
	A common choice is  $w_k=|\mathcal{S}_k|/\sum_{j\in\Pi(t)}|\mathcal{S}_j|$.
\end{enumerate}
In FedAvg, the server needs to broadcast the global model at each iteration. 
The model parameters $\boldsymbol{\theta}(t)$ progress in the gradient descent directions. This naturally introduces temporal correlation between consecutive global models. In Figure~\ref{fig:expGrdt}, we illustrate this temporal correlation in the global model evolution when using FedAvg with a convolutional neural network for an image classification problem.\footnote{We consider the same learning setting and use the same parameters as in the other simulation results presented in Section \ref{sec:simulation_results}.} In this example,  high correlation is observed even when two iterates are $20$ iterations apart. In general, the correlation between two iterates weakens when the time window increases, and it increases with a smaller learning rate.
%In the FL application, there is no explicit memory model, and our proposed techniques, which we will introduce later, do not rely on any specific model
%for how the iterates are correlated.}
This observation motivates the usage of DiC techniques to improve DL communication efficiency. In the literature, DiC has been considered for UL local update transmission \cite{adrian2024Temp, yue2022com}, but to our knowledge it has not been investigated for DL global model transmission.

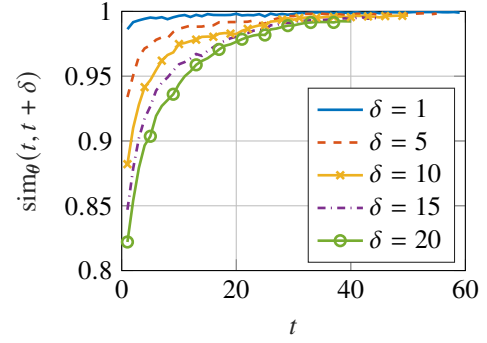
\begin{figure}[t!]
    \centering
	%\hspace{2cm}
	%\vspace{-4.5cm}
\definecolor{mycolor1}{rgb}{0.00000,0.44700,0.74100}%
\definecolor{mycolor2}{rgb}{0.85000,0.32500,0.09800}%
\definecolor{mycolor3}{rgb}{0.92900,0.69400,0.12500}%
\definecolor{mycolor4}{rgb}{0.49400,0.18400,0.55600}%
\definecolor{mycolor5}{rgb}{0.46600,0.67400,0.18800}%
\begin{tikzpicture}

\begin{axis}[%
width=1.79in,
height=1.35in,
at={(0.763in,0.64in)},
scale only axis,
xmin=0,
xmax=60,
xlabel style={font=\color{white!15!black}},
xlabel={$t$},
ymin=0.8,
ymax=1,
ylabel style={font=\color{white!15!black}},
ylabel={sim$_{\boldsymbol{\theta}}(t,t+\delta)$},
axis background/.style={fill=white},
xmajorgrids,
ymajorgrids,
legend style={at={(0.97,0.03)}, anchor=south east, legend cell align=left, align=left, draw=white!15!black}
]
\addplot [color=mycolor1, line width=1.0pt]
  table[row sep=crcr]{%
1	0.98624152\\
2	0.99167764\\
3	0.99309629\\
4	0.99426562\\
5	0.99522942\\
6	0.99470037\\
7	0.99548656\\
8	0.99395162\\
9	0.99544042\\
10	0.99656016\\
11	0.99716544\\
12	0.99566007\\
13	0.99732137\\
14	0.99623483\\
15	0.99787223\\
16	0.99722975\\
17	0.99716902\\
18	0.99713367\\
19	0.99726188\\
20	0.99813259\\
21	0.99686736\\
22	0.9972046\\
23	0.99698836\\
24	0.9978314\\
25	0.99612254\\
26	0.99821687\\
27	0.99728775\\
28	0.99800062\\
29	0.99741066\\
30	0.99801159\\
31	0.9988265\\
32	0.99859828\\
33	0.99905956\\
34	0.99920708\\
35	0.99911028\\
36	0.99937999\\
37	0.99855328\\
38	0.99928743\\
39	0.99893069\\
40	0.99883628\\
41	0.99939197\\
42	0.99858785\\
43	0.99873638\\
44	0.99927956\\
45	0.99857694\\
46	0.99921423\\
47	0.99931425\\
48	0.99925655\\
49	0.99915767\\
50	0.99919122\\
51	0.99948788\\
52	0.99944246\\
53	0.99944448\\
54	0.99946618\\
55	0.99959946\\
56	0.99932754\\
57	0.99928981\\
58	0.9994756\\
59	0.99898523\\
};
\addlegendentry{$\delta=1$}

\addplot [color=mycolor2, dashed, line width=1.0pt]
  table[row sep=crcr]{%
1	0.93387663\\
2	0.95051515\\
3	0.96511334\\
4	0.97141486\\
5	0.97393668\\
6	0.9780491\\
7	0.97995859\\
8	0.97941536\\
9	0.98408097\\
10	0.98650461\\
11	0.98776156\\
12	0.98901904\\
13	0.98818225\\
14	0.98826736\\
15	0.98913658\\
16	0.98924404\\
17	0.99157876\\
18	0.99160576\\
19	0.99175745\\
20	0.99175966\\
21	0.99184787\\
22	0.99201262\\
23	0.99293935\\
24	0.99102026\\
25	0.99053609\\
26	0.99296606\\
27	0.9943338\\
28	0.99522609\\
29	0.99643439\\
30	0.99633604\\
31	0.99714613\\
32	0.99806952\\
33	0.99812025\\
34	0.99818081\\
35	0.99813753\\
36	0.99784964\\
37	0.99722421\\
38	0.99775738\\
39	0.99695367\\
40	0.99731892\\
41	0.9980908\\
42	0.99812376\\
43	0.99772352\\
44	0.99794602\\
45	0.99711847\\
46	0.99816096\\
47	0.9974637\\
48	0.99800122\\
49	0.99793059\\
50	0.99799275\\
51	0.99816877\\
52	0.99838036\\
53	0.99857825\\
54	0.99778205\\
55	0.99830186\\
};
\addlegendentry{$\delta=5$}

\addplot [color=mycolor3, line width=1.0pt, mark=x, mark options={solid, mycolor3}, mark repeat=3]
  table[row sep=crcr]{%
1	0.88210917\\
2	0.91005546\\
3	0.92801249\\
4	0.94152415\\
5	0.94693828\\
6	0.95471174\\
7	0.96196949\\
8	0.96630985\\
9	0.96833152\\
10	0.97478908\\
11	0.97560966\\
12	0.97699279\\
13	0.97818106\\
14	0.98034149\\
15	0.98053938\\
16	0.98040217\\
17	0.98276091\\
18	0.98327494\\
19	0.98282588\\
20	0.98250443\\
21	0.98498517\\
22	0.98648065\\
23	0.9880892\\
24	0.9891693\\
25	0.98907375\\
26	0.99210119\\
27	0.99252099\\
28	0.99367589\\
29	0.99482435\\
30	0.99355012\\
31	0.99494183\\
32	0.99515599\\
33	0.99573857\\
34	0.99522722\\
35	0.99535757\\
36	0.99590999\\
37	0.99577522\\
38	0.99595803\\
39	0.99569714\\
40	0.99555367\\
41	0.99651206\\
42	0.99648029\\
43	0.99611497\\
44	0.99590349\\
45	0.99535114\\
46	0.99635994\\
47	0.99648315\\
48	0.99666542\\
49	0.9965778\\
50	0.99677694\\
};
\addlegendentry{$\delta=10$}

\addplot [color=mycolor4, dashdotted, line width=1.0pt]
  table[row sep=crcr]{%
1	0.84690833\\
2	0.87875676\\
3	0.90214044\\
4	0.91741502\\
5	0.92673278\\
6	0.93881589\\
7	0.9449231\\
8	0.94898826\\
9	0.95427519\\
10	0.95916653\\
11	0.96089602\\
12	0.96365094\\
13	0.96710867\\
14	0.96568167\\
15	0.96958321\\
16	0.97238106\\
17	0.97534436\\
18	0.97728199\\
19	0.97924763\\
20	0.9804464\\
21	0.9826405\\
22	0.98326015\\
23	0.98509479\\
24	0.9859708\\
25	0.98552096\\
26	0.98932022\\
27	0.98947358\\
28	0.99080575\\
29	0.99118787\\
30	0.9912889\\
31	0.9926455\\
32	0.99279785\\
33	0.9931885\\
34	0.9935469\\
35	0.99324191\\
36	0.99363631\\
37	0.99354905\\
38	0.9944098\\
39	0.99460387\\
40	0.99390453\\
41	0.99503142\\
42	0.99539536\\
43	0.9943682\\
44	0.99500406\\
45	0.99430335\\
};
\addlegendentry{$\delta=15$}

\addplot [color=mycolor5, line width=1.0pt, mark=o, mark options={solid, mycolor5}, mark repeat=4]
  table[row sep=crcr]{%
1	0.82206476\\
2	0.85393441\\
3	0.87927598\\
4	0.89658564\\
5	0.90362626\\
6	0.91934633\\
7	0.92701137\\
8	0.93152803\\
9	0.93605226\\
10	0.94387591\\
11	0.94988221\\
12	0.95373541\\
13	0.95874476\\
14	0.96174562\\
15	0.9638263\\
16	0.96695745\\
17	0.97063261\\
18	0.97343773\\
19	0.97459531\\
20	0.97652209\\
21	0.97850651\\
22	0.98010892\\
23	0.98137075\\
24	0.98207217\\
25	0.98174274\\
26	0.98606706\\
27	0.98671252\\
28	0.98790419\\
29	0.98899984\\
30	0.98919231\\
31	0.9901731\\
32	0.991072\\
33	0.991413\\
34	0.99178296\\
35	0.99179029\\
36	0.99175775\\
37	0.99147421\\
38	0.99243838\\
39	0.99231386\\
40	0.99224275\\
};
\addlegendentry{$\delta=20$}

\end{axis}
\end{tikzpicture}%
	\caption{
    Illustration of the temporal correlation of the iterates $\boldsymbol{\theta}(t)$ quantified by the cosine similarity between $\boldsymbol{\theta}(t)$ and $\boldsymbol{\theta}(t+
    \delta)$, denoted by sim$_{\boldsymbol{\theta}}(t,t+\delta)$, where $\delta=1,5,10,15,20$.}
	\label{fig:expGrdt}
\end{figure}

\subsection{Differential Coding}
\label{sec:df}
%\zheng{Give some references of differential/predictive coding for image compression. You can even use $x_t$ to denote the image frame so that the explanation of the scheme is easier. }
Differential coding (DiC) is a common technique for multimedia data compression \cite{sayood2000introduction}. 
For instance, let $\boldsymbol{x}_t$ be the current image frame and $\boldsymbol{x}_{\text{ref}}$ be the reference frame.  
The sender transmits the compressed residual $\triangle\boldsymbol{x}_t=Q(\boldsymbol{x}_t-\boldsymbol{x}_{\textnormal{ref}})$, where $Q(\cdot)$ is a compression operator (e.g., quantizer). We will call $\triangle\boldsymbol{x}_t$ the differential update throughout the paper.
Let $\tilde{\boldsymbol{x}}_{t}$ be the reconstructed frame at the receiver side, which can be computed by $\tilde{\boldsymbol{x}}_{t}=\boldsymbol{x}_{\textnormal{ref}}+\triangle\boldsymbol{x}_t$ at the receiver. To avoid error propagation in $\tilde{\boldsymbol{x}}_{t}$, the residual calculation will use the reference frame $\boldsymbol{x}_{\textnormal{ref}}$ based on $\{\tilde{\boldsymbol{x}}_i\}_{i<t}$ rather than the original frames $\{\boldsymbol{x}_i\}_{i<t}$. A special case is  $\boldsymbol{x}_{\textnormal{ref}}=\tilde{\boldsymbol{x}}_{t-1}$, the latest reconstructed frame.
%\zheng{Explain the advantages of DiC: smaller signal magnitude, smaller dynamic range, need less quantization bits to reach similar MSE level.}

DiC offers the advantage of reducing the signal magnitude and dynamic range, thereby requiring fewer quantization bits to achieve a similar mean-squared error (MSE)   as the full model transmission.
However, DiC schemes rely on perfect reception of the differential update $\triangle\boldsymbol{x}_t$ and the reconstructed reference frame.
%\zheng{Current image frame be $x_t$, reference/predicted image frame $\hat{x}_t$, transmit residual as the difference between the current and predicted image. Special cases, predicted image can be the reconstructed previous image frame. Rely on two components: prediction and perfect reception of residual.}
%A reasonable way to deploy differential coding in FL is to broadcast a full model occasionally ($\boldsymbol{x}_t$ for some $t$) and for the rest a differential model ($\triangle\boldsymbol{x}_t$ otherwise). 
%One major defect of differential coding is vulnerability to decoding failures. 
Missing a differential update $\triangle\boldsymbol{x}_t$ makes the reconstruction of the current and all the subsequent frames (i.e., $\{\tilde{\boldsymbol{x}}_i\}_{i\geq t}$) impossible. 
Wireless communication  is generally susceptible to random decoding failures, and in the context of FL, missing a model update will cause a device to work on an outdated global model
when performing 
%. More specifically, the device will use an outdated global model to conduct 
local training. One way to tackle this issue is to occasionally schedule full-model broadcasts along the   process.
An example is illustrated in Figure \ref{fig:broadcastTimeDic}, where at any iteration $t$, $\hat{\boldsymbol{\theta}}(t)$ is the transmitted vector (either a full model or a differential update) and $\tilde{\boldsymbol{\theta}}(t)$ is the reconstructed model with respect to $\boldsymbol{\theta}(t)$. 
Depending on whether $\hat{\boldsymbol{\theta}}(t)$ is successfully received at device $k$, the local model before training, $\boldsymbol{\theta}_k(t,0)$, can be either $\tilde{\boldsymbol{\theta}}(t)$ or an outdated model $\boldsymbol{\theta}_k(t-1,0)$. The infrequent full-model broadcasts offer the possibility to re-synchronize to the latest global model, but with an increased communication cost compared to only transmitting differential updates.
%\zheng{Full model doesn't cost "extra" communication resource compared to FedAvg. It costs extra compared to FedAvg with DiC in every iteration. Also, explain why transmitting full models consume more resources, or conversely, why transmitting differential updates consume less resources.}
Bearing this in mind, we propose the following mixed-timescale differential coding (MTDC) scheme.
\begin{figure}[t!]
	\centering
	\includegraphics[scale=0.45]{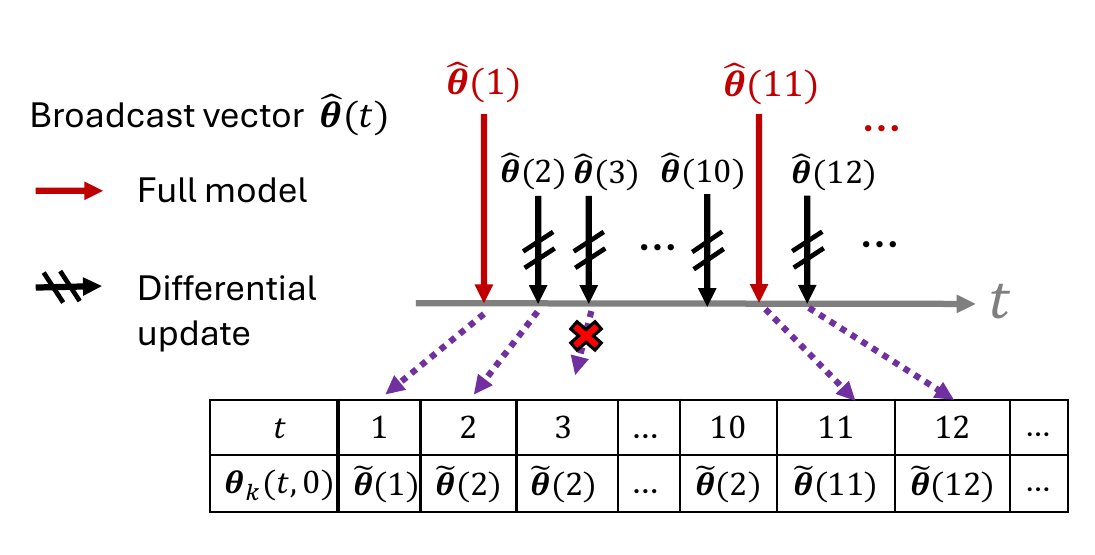}
	\caption{An example of FL with DiC-enabled model broadcast: the server broadcasts the full model every $10$ iterations and in between differential updates are broadcast instead.}
	\label{fig:broadcastTimeDic}
\end{figure}

\subsection{Mixed-Timescale Differential Coding}
\label{sec:mtdc_intro}
% \egl{In case we need to shorten, this paragraph could probably be taken out
% (at a high level, how MTDC works is already explained in the introduction)}
We consider DiC-enabled FL model broadcast with one reference model.\footnote{It can be extended to more advanced prediction schemes using multiple reference models reconstructed at different time instances.} 
%To tackle the decoding failure issue, we  
Other than assigning the latest reconstructed model as the reference (e.g., DiC in Sec. \ref{sec:df}), MTDC offers flexibility in choosing references from different past instances in the training process.
%\zheng{In our work, we use reference model as the reconstructed previous model. Add footnote that it can extended to other prediction schemes and the reference model can be calculated in other ways using previous information.}
In this work, we focus on MTDC with two-level differential updates.
%\footnote{It can be conceptually extended to the case with more-level differential models.} 

We first introduce the iteration index sets $\{\mathcal{T}_i\}_{i=0}^2$. When $t\in\mathcal{T}_0$, the server broadcasts a full model, and when $t\in\mathcal{T}_1$ and $t\in\mathcal{T}_2$, a first-level and a second-level differential updates are broadcast, respectively. All broadcast vectors from the server are compressed before being sent. We denote the compression functions by $\{Q_i(\boldsymbol{\theta})\}_{i=0}^2$ when $t\in\mathcal{T}_i,i=0,1,2$.
Let $\hat{\boldsymbol{\theta}}(t)$ be the transmitted vector and $\tilde{\boldsymbol{\theta}}(t)$ be the reconstructed model with respect to $\boldsymbol{\theta}(t)$. 
When the full model is broadcast, $\hat{\boldsymbol{\theta}}(t)=Q_0(\boldsymbol{\theta}(t))$. At other times, $\hat{\boldsymbol{\theta}}(t)$ is the corresponding differential
update.
Then, in lieu of \eqref{eq:syncFlAggregation}, the updating rule becomes\footnote{We assume perfect reception of $\triangle\boldsymbol{\theta}_k(t)$ by the server.}
\begin{equation}
	\boldsymbol{\theta}(t+1)=\tilde{\boldsymbol{\theta}}(t)+\sum_{ k\in\Pi(t)}w_k\triangle\boldsymbol{\theta}_k(t),\label{eq:update_MTDC}
\end{equation}
with $\triangle\boldsymbol{\theta}_k(t)$ defined in \eqref{eq:triThetak} and $\boldsymbol{\theta}_k(t,0)=\tilde{\boldsymbol{\theta}}(t)$, to keep the models at the server and the devices synchronized.
%\zheng{The DiC transmission is in the DL. Why should the UL aggregation rule been updated? To keep the global model synchronized between the server and the devices? This should be clarified.}
%\zheng{Add in footnote that here we assume $\Delta \theta_k$ are perfected received.}

\begin{figure*}[t!]
	\centering
	\includegraphics[scale=0.5]{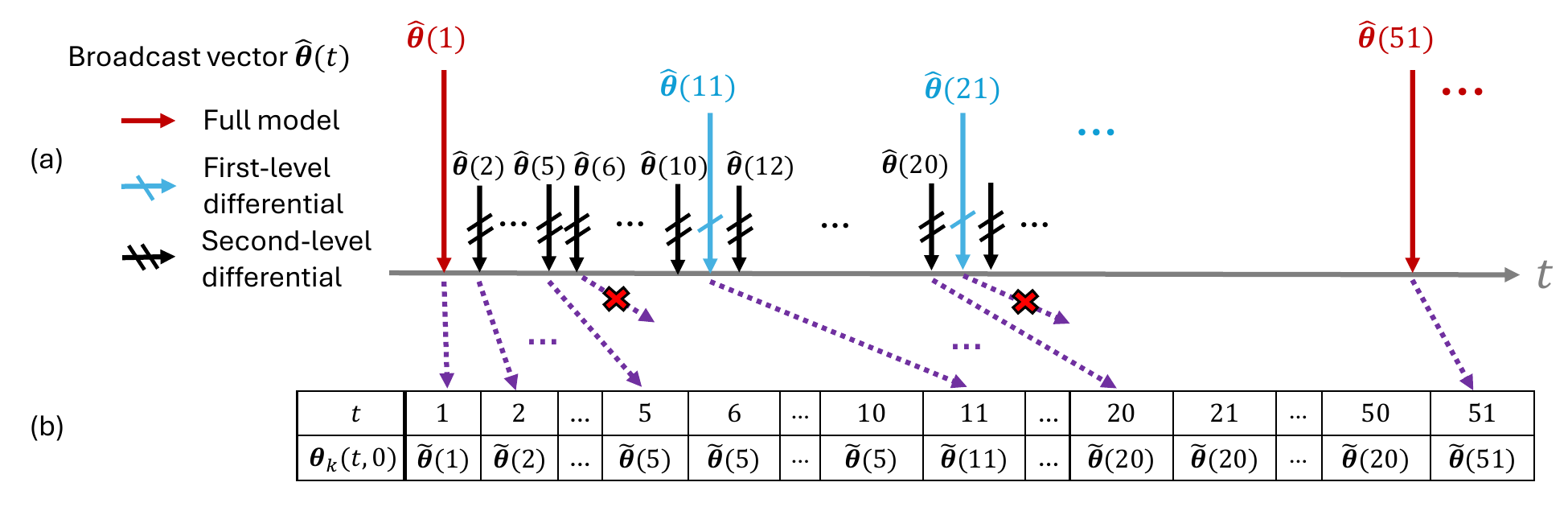}
	\caption{(a) The server broadcasts a full model at times $1,51,101,...$, a first-level differential update at times $11,21,...$, and a second-level differential update at times $2,...,10$, $12,...,20,...$. (b) Reception of $\hat{\boldsymbol{\theta}}(t)$ and possession of the corresponding reference model determine the adopted model at device $k$, $\boldsymbol{\theta}_k(t,0)$, to be either $\tilde{\boldsymbol{\theta}}(t)$ or $\boldsymbol{\theta}_k(t-1,0)$.}
	\label{fig:broadcastTime}
\end{figure*}

An illustrative example with $\mathcal{T}_0=\{1,51,...\}$, $\mathcal{T}_1=\{11,21,...\}$, and $\mathcal{T}_2=\{2,3,...\}$ is in Fig. \ref{fig:broadcastTime}a:
\begin{itemize}
	\item $\hat{\boldsymbol{\theta}}(1)$ and $\hat{\boldsymbol{\theta}}(51)$ are the compressed full models corresponding to $\boldsymbol{\theta}(1)$ and $\boldsymbol{\theta}(51)$. Thus, $\hat{\boldsymbol{\theta}}(1)=Q_0(\boldsymbol{\theta}(1))$ and $\hat{\boldsymbol{\theta}}(51)=Q_0(\boldsymbol{\theta}(51))$.
	\item A first-level differential update $\hat{\boldsymbol{\theta}}(11)$ carries the difference between $\boldsymbol{\theta}(11)$ and $\tilde{\boldsymbol{\theta}}(1)$ after compression, i.e., $Q_1(\boldsymbol{\theta}(11)-\tilde{\boldsymbol{\theta}}(1))$. Similarly, $\hat{\boldsymbol{\theta}}(21)=Q_1(\boldsymbol{\theta}(21)-\tilde{\boldsymbol{\theta}}(11))$. 
	\item A second-level differential update carries the difference between $\boldsymbol{\theta}(t)$ and $\tilde{\boldsymbol{\theta}}(t-1)$ after compression, e.g., $\hat{\boldsymbol{\theta}}(2)=Q_2(\boldsymbol{\theta}(2)-\tilde{\boldsymbol{\theta}}(1))$ and $\hat{\boldsymbol{\theta}}(6)=Q_2(\boldsymbol{\theta}(6)-\tilde{\boldsymbol{\theta}}(5))$.
\end{itemize}
Wireless transmissions are susceptible to decoding failures. 
We define the decoding failure probability of device $k$ at iteration $t$ as $P_{k,i}^{(t)}$, with $t\in\mathcal{T}_i$, $i\in\{0,1,2\}$.
To ensure  successful reconstruction, the transmissions of $\hat{\boldsymbol{\theta}}(t), t\in\mathcal{T}_i$, $i\in\{0,1,2\}$ are coded such that
$P^{(t)}_{k,0}\ll P^{(t)}_{k,1}\leq P^{(t)}_{k,2}$.\footnote{Making full-model broadcasts available to all facilitates device participation in model evolution under differential coding schemes.}
%\zheng{Need to be revised}
%\zheng{"Code rate" has different meanings in source coding and channel coding. Be specific about what you mean here.}
%\zheng{reception with higher decoding success probability or reception with higher accuracy? This connects to my question about the meaning of "code rate".}
When a device fails to decode $\hat{\boldsymbol{\theta}}(t)$, the local training is based on an outdated model $\tilde{\boldsymbol{\theta}}(\tau)$, for some $\tau<t$.
% which generally leads to slow convergence.
The proposed scheme enables   re-synchronization with the server, improving model staleness in case of decoding failures.   Fig. \ref{fig:broadcastTime}b
exemplifies:
\begin{itemize}
	\item When $\hat{\boldsymbol{\theta}}(t)$ is decoded successfully by device $k$ and the reference model is available (e.g., at $t=2,11$, the reference $\tilde{\boldsymbol{\theta}}(1)$ is available), the adopted model for local training is up-to-date, i.e., $\boldsymbol{\theta}_k(t,0)=\tilde{\boldsymbol{\theta}}(t)$.
	\item When device $k$ misses $\hat{\boldsymbol{\theta}}(6)$, it has to use $\tilde{\boldsymbol{\theta}}(5)$, the model reconstructed at $t=5$ as the adopted model, until $t=11$, at which point it can rely on the first-level differential $\hat{\boldsymbol{\theta}}(11)$ and the memory of $\tilde{\boldsymbol{\theta}}(1)$ to compute $\tilde{\boldsymbol{\theta}}(11)$.
	\item When a device misses the transmitted information $\hat{\boldsymbol{\theta}}(21)$, it has to adopt the reconstructed model $\tilde{\boldsymbol{\theta}}(20)$ until the next full-model broadcast at $t=51$.
\end{itemize}
Compared to the baseline DiC where $\boldsymbol{\theta}_k(50,0)=\tilde{\boldsymbol{\theta}}(5)$ (due to decoding failure of $\hat{\boldsymbol{\theta}}(6)$), the proposed MTDC improves the staleness with $\boldsymbol{\theta}_k(50,0)=\tilde{\boldsymbol{\theta}}(20)$.

\begin{remark}
Note that with transmitting either full models or differentially coded updates, the communication frequency is the same; the difference lies in what the communicated message contains.
When successive global models are highly correlated, differential updates exhibit a much smaller magnitude and dynamic range than full models, thus requiring fewer quantization bits.
\end{remark}

In the following sections, we explain in detail how MTDC is implemented (Sec. \ref{sec:flmtdc}). Then, we provide a   convergence analysis (Sec. \ref{sec:conv_study}), which motivates the age-aware MTDC and scheduling design proposed in Sec. \ref{sec:age-aware-design}.
Table \ref{tab:param} summarizes the  notation.
%\zheng{The effects of decoding failures should be emphasized.}
\begin{figure*}[t!]
	\centering
	\includegraphics[scale=0.6]{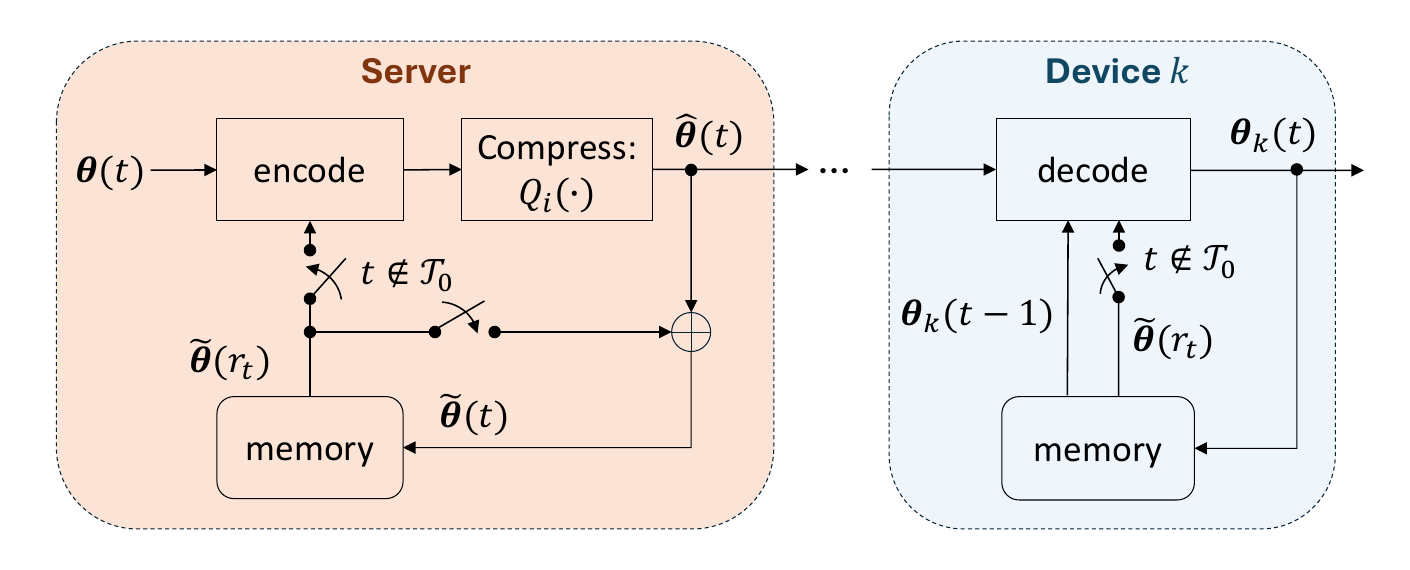}
	\caption{Block diagram of an FL system with MTDC.}
	\label{fig:blockDiagram}
\end{figure*}
\begin{table}[t]
	\caption{Parameter definitions at iteration $t$} %title of the table
	\centering % centering table
	\normalsize
	\begin{tabular}{|c|c|} % creating eight columns
		\hline
		\textbf{Parameter}& \textbf{Definition}\\
		\hline
		%\\[1ex] % inserts single-line
		$\boldsymbol{\theta}(t)$/$\tilde{\boldsymbol{\theta}}(t)$ & true/reconstructed model\\
        $\hat{\boldsymbol{\theta}}(t)$&transmitted vector\\
		$Q_i(\boldsymbol{\theta})$ & compression function for $t\in\mathcal{T}_i,i=0,1,2$\\
		  $\triangle\boldsymbol{\theta}_k(t)$&model update from device $k$\\
          \multirow{2}{*}{$P_{k,i}^{(t)}$}&decoding failure probability of\\
        &device $k$ for $t\in\mathcal{T}_i,i=0,1,2$\\
		\multirow{2}{*}{$\boldsymbol{\theta}_k(t)$} & adopted model before local training\\
        &at device $k$ (i.e., $\boldsymbol{\theta}_k(t)=\boldsymbol{\theta}_k(t,0)$)\\
        $r_t$& timestamp of reference model. See \eqref{eq:t_ref}\\
		$a_k^{(t)}$& age of $\boldsymbol{\theta}_k(t)$\\%[1ex] % [1ex] adds vertical space
		\hline % inserts single-line
	\end{tabular}
	\label{tab:param}
\end{table} 

\section{Federated Learning with Mixed-Timescale Differential Coding}
\label{sec:flmtdc}
For an FL system with MTDC, at iteration $t$, $\tilde{\boldsymbol{\theta}}(\tau),\tau\leq t,$ is the global model that a device has worked on in order to compute its local update. This model  also serves as the base model for the server to compute the broadcast for some of the future iterations
(differential updates).
To simplify notation, we denote the locally adopted model before local training at device $k$ by $\boldsymbol{\theta}_k(t)$ (i.e., $\boldsymbol{\theta}_k(t)=\boldsymbol{\theta}_k(t,0)$).
%, $\hat{\boldsymbol{\theta}}(\tau),\tau>t$.
%The block diagram in 
Fig. \ref{fig:blockDiagram} explains:
\begin{itemize}
	\item how the server computes the broadcast vector $\hat{\boldsymbol{\theta}}(t)$ and
    the reconstructed model $\tilde{\boldsymbol{\theta}}(t)$, and;
	\item how device $k$ computes the adopted model $\boldsymbol{\theta}_k(t)$, based on the received  $\hat{\boldsymbol{\theta}}(t)$ and on knowledge of previous reconstructed models (e.g., $\tilde{\boldsymbol{\theta}}(\tau)$ for some $\tau<t$), 
\end{itemize}
at any iteration $t$. More details are elaborated as follows.

\subsection{Server-Side Operations}

\label{sec:srv-opt}
As illustrated in Fig. \ref{fig:blockDiagram}, the broadcast vector at iteration $t$ is
\begin{equation}
	\hat{\boldsymbol{\theta}}(t)=\begin{cases}
		Q_0(\boldsymbol{\theta}(t)), & t\in\mathcal{T}_0\\
		Q_i(\boldsymbol{\theta}(t)-\tilde{\boldsymbol{\theta}}(r_t)), & t\in\mathcal{T}_i,i=1,2
	\end{cases},\label{eq:h_hat_def}
\end{equation}
where $r_t$ is the timestamp of the reference model
\begin{equation}
	r_t=\begin{cases}
		\max\{\tau|\tau\in\mathcal{T}_0\bigcup \mathcal{T}_1,\tau<t\}, & t\in\mathcal{T}_1\\
		t-1, & t\in\mathcal{T}_2 
		\label{eq:t_ref}
	\end{cases},
\end{equation}
being either the last full-model broadcast or  first-level differential broadcast (for $t\in\mathcal{T}_1$),
or the model at the previous iteration (for $t\in\mathcal{T}_2 $).

We consider $\nu_i$-level random quantizers,\footnote{
We adopt the random quantizer for its unbiasedness property.
    The quantization noise has zero mean,
    conditioned  on the iterate, which facilitates a convergence analysis of the proposed algorithm.} for the different levels of differential coding, $i=0,1,2$ \cite{Alistarh2017QSGDCS}. 
That is, for $i\in\{1,2,3\}$, $\boldsymbol{\theta}\triangleq[x_1,...,x_d]$, the $j$-th element of $Q_i\left(\boldsymbol{\theta}\right)$ is
\begin{equation}
	\|\boldsymbol{\theta}\|_2\cdot\text{sign}(x_j)\cdot \mathcal{Z}_j(\boldsymbol{\theta},\nu_i),
\end{equation}
where
\begin{equation}
	\label{eq:zj}\mathcal{Z}_j(\boldsymbol{\theta},\nu_i)=\begin{cases}
		\left(X_j+1\right)/\nu_i,&\text{ with prob. }\nu_i|X_j|/\|\boldsymbol{\theta}\|_2-X_j\\
		X_j/\nu_i,&\text{ otherwise}
	\end{cases}
\end{equation}
is a random variable and $X_j=\lfloor\nu_i|x_j|/\|\boldsymbol{\theta}\|_2\rfloor$, $j=1,...,d$. Clearly, $Q_i(\boldsymbol{0})=\boldsymbol{0}$, and each pair in the following cases has the same probability distribution: (1) $Q_i(\eta\boldsymbol{\theta})$ and $\eta Q_i(\boldsymbol{\theta})$ for any $\eta>0$; (2) $Q_i(\boldsymbol{\theta})$ and $-Q_i(-\boldsymbol{\theta})$. 
For every transmission of $\boldsymbol{\theta}$, the quantization step is dynamically adjusted as $\|\boldsymbol{\theta}\|/\nu_i$.
$\mathcal{Z}_j(\cdot,\cdot), \forall j$, are real numbers taken from a finite set of at most $\nu_i+1$ rational numbers. This set of rational numbers in turn can be mapped one-to-one onto a finite set of integers.
\begin{remark}
Note that each $\nu_i$-level random quantizer $Q_i(\boldsymbol{\theta}), i=0,1,2,$ has the properties
	\begin{equation}
		\mathbb{E}\left[Q_i(\boldsymbol{\theta})|\boldsymbol{\theta}\right]=\boldsymbol{\theta},~~\mathbb{E}\left[\|\boldsymbol{\theta}-Q_i(\boldsymbol{\theta})\|_2^2|\boldsymbol{\theta}\right]\leq\sigma_i\|\boldsymbol{\theta}\|_2^2,\label{eq:rdmqnt}
	\end{equation}	
	where a higher $\nu_i$ gives a smaller precision constant $\sigma_i$. Since a full model, a first- and a second-level differential updates tend to have the highest to the lowest signal ranges respectively, the required $\nu_i$ for a sufficient signal quality decreases over $i$. This leads to $\sigma_2\geq\sigma_1\geq\sigma_0$.
	\label{rmk:rdmqnt}
\end{remark}
\begin{remark}
\label{rmk:qnt_bits}
For $t\in\mathcal{T}_i, \forall i$, the transmission of the $\nu_i$-level-quantized $\hat{\boldsymbol{\theta
}}(t)$ requires $d\left(\lceil\log_2(\nu_i+1)\rceil+1\right)+32$ bits.    
\end{remark}

The server computes the reconstructed global model by 
\begin{equation}
	\tilde{\boldsymbol{\theta}}(t)=\begin{cases}
		\hat{\boldsymbol{\theta}}(t), & t\in\mathcal{T}_0\\
		\hat{\boldsymbol{\theta}}(t)+\tilde{\boldsymbol{\theta}}(r_t), & \text{otherwise,}
	\end{cases}
	\label{eq:mdl_rcst}
\end{equation}
which is then saved in the memory of the server. This reconstructed model, $\tilde{\boldsymbol{\theta}}(t)$, serves as a reference model for future differential broadcasts.
%The broadcast model $\hat{\boldsymbol{\theta}}(t)$ can be computed as
%\begin{equation*}
%	\hat{\boldsymbol{\theta}}(t)=\mathcal{B}\left(\triangle\boldsymbol{\theta}(t-1),\{\boldsymbol{\theta}^{(i)}\}_{i=0}^{2}\right)+\boldsymbol{e}(t)
%	%\label{eq:bcst_spf}
%\end{equation*}
%for some function $\mathcal{B}$ and compression noise $\boldsymbol{e}(t)$.
%$\triangle\boldsymbol{\theta}(t-1)$ is the aggregated model update from the previous iteration, and $\{\boldsymbol{\theta}^{(i)}\}_{i=0}^{2}$ are the stored memory for the last full, accumulated first-level and second-level differential models respectively.

\subsection{Device-Side Operations}
\label{sec:dvc-opt}
Device $k$ relies on  $\hat{\boldsymbol{\theta}}(t)$  broadcast by the server, together with the reference model reconstructed at an earlier iteration, to compute the adopted model $\boldsymbol{\theta}_k(t)$ before local training. That is,
\begin{equation}
	\boldsymbol{\theta}_k(t)=\begin{cases}
		\tilde{\boldsymbol{\theta}}(t),&\text{$\hat{\boldsymbol{\theta}}(t)$ received, $t\in\mathcal{T}_0$; or}\\
		&\text{$\hat{\boldsymbol{\theta}}(t)$ received, $t\notin\mathcal{T}_0$, and $\tilde{\boldsymbol{\theta}}(r_t)$}\text{ available}\\
		\boldsymbol{\theta}_k(t-1),&\text{otherwise}
	\end{cases}
	\label{eq:adopted_model}
\end{equation}
where  $r_t$ and $\tilde{\boldsymbol{\theta}}(t)$ are defined in  \eqref{eq:t_ref} and \eqref{eq:mdl_rcst}, respectively.
A device has $\tilde{\boldsymbol{\theta}}(r_t)$ in its memory only  when it has successfully reconstructed it in iteration $r_t$.
If $\boldsymbol{\theta}_k(t)=\tilde{\boldsymbol{\theta}}(t)$, corresponding to a successful reconstruction of the latest model, device $k$ saves $\tilde{\boldsymbol{\theta}}(t)$ in its memory for future computation. $\boldsymbol{\theta}_k(t)$ will be saved as well in case of a future model reconstruction failure.
%Specifically, at iteration $t$, for each device $k$, such a reconstructed model can be expressed as
%\begin{equation*}
%	\tilde{\boldsymbol{\theta}}_k(t)=\mathcal{R}\left(\hat{\boldsymbol{\theta}}(t),\{\boldsymbol{\theta}_k^{(i)}\}_{i=0}^{2}\right)
%	%\label{eq:dev_rct}
%\end{equation*}
%for some function $\mathcal{R}$. The memory buffers $\{\boldsymbol{\theta}_k^{(i)}\}_{i=0}^{2}$ keep track of the latest broadcast full model, and the accumulated model differentials of first-level and second-level respectively.
%A device can reconstruct the latest model only if it receives $\hat{\boldsymbol{\theta}}(t)$ successfully. 

%Following \eqref{eq:adopted_model}, if a device misses a second-level update, it can only use the last known model until the next first-level or full model update. 
%Similarly, missing a first-level update results in adopting an outdated model until the next full model broadcast. 
\subsection{Discussion on Communication and Memory Overhead}
\label{rmk:fairComp}
    Compared to the full model broadcast, the two-level MTDC requires additional communication  of $\lceil\log_2 (t-r_t)\rceil$ and $2$ bits, for transmitting the timestamp of the reference model, and to convey the broadcast model type information in each transmission block. This extra signaling overhead is negligible compared to the maximum payload size of a data packet. Also, the differential-coding-based schemes effectively improve the communication efficiency, but require additional memory for storing the reference model.
\begin{remark}\label{rmk:deFLdisc}
In principle, the MTDC mechanism is applicable also to decentralized FL frameworks \cite{xing2021flo,yan2024perf}.
However, the modeling would be more involved as the links between different devices may have different quality and fail independently, and each node needs to keep track of historical model information for all its neighbors. 
Note that MTDC may not be directly applicable to decentralized FL with over-the-air computation, as in this case the aggregation of analog signals makes it impossible to distinguish and track individual local models separately.
\end{remark}

\section{Convergence Analysis}
\label{sec:conv_study}

In our MTDC FL system, some devices may conduct the local training (gradient computation) based on outdated models.
The question is then, whether convergence of the learning algorithm can be guaranteed.
In the following analysis, we answer this question affirmatively under standard assumptions on the objectives and some idealized additional assumptions on the model:
%We conduct theoretic analyses for an FL system with MTDC. To focus on the effect of device decoding failures and broadcast model quantization on the convergence of learning performance, we provide the analyses in simpler scenarios where 
no device scheduling ($\Pi(t)=\mathcal{K}$), 
a single local gradient step per iteration ($E=1$),
no sampling noise in the gradient computation ($\mathcal{B}_k(t,0)=\mathcal{S}_k,\forall k$), and no quantization or communication noise in the uplink gradient transmissions. 
In more detail, these assumptions are as follows.
%The essential assumptions, system properties and parameters are introduced as follows.
\begin{assumption}
	(Smoothness): Each local loss function $F_k(\boldsymbol{\theta}),\forall k$ is $L$-smooth, i.e., $\forall \boldsymbol{\theta}_1,\boldsymbol{\theta}_2\in\mathbb{R}^d$,\footnote{$\|\cdot\|$ denotes the Euclidean norm.}
	\begin{equation*}
		\|\nabla F_k\left(\boldsymbol{\theta}_1\right)-\nabla F_k\left(\boldsymbol{\theta}_2\right)\|\leq{L}\|\boldsymbol{\theta}_1-\boldsymbol{\theta}_2\|,
	\end{equation*}
	or equivalently,
	\begin{equation*}
		F_k(\boldsymbol{\theta}_1)-F_k(\boldsymbol{\theta}_2)\leq\nabla F_k(\boldsymbol{\theta}_2)^T(\boldsymbol{\theta}_1-\boldsymbol{\theta}_2)+\frac{L}{2}\|\boldsymbol{\theta}_1-\boldsymbol{\theta}_2\|^2.
	\end{equation*}
	\label{asp:smth}	
\end{assumption}
\begin{assumption}
	(Strong convexity): Each local loss function $F_k(\boldsymbol{\theta}),\forall k$ is $\mu$-strongly convex, i.e., $\forall \boldsymbol{\theta}_1, \boldsymbol{\theta}_2\in\mathbb{R}^d$,
	\begin{equation*}
		%\label{eq:asmp2_localSGD}
		F_k\left(\boldsymbol{\theta}_1\right)-F_k\left(\boldsymbol{\theta}_2\right)\geq\nabla F_k\left(\boldsymbol{\theta}_2\right)^T\left(\boldsymbol{\theta}_1-\boldsymbol{\theta}_2\right)+\frac{\mu}{2}\|\boldsymbol{\theta}_1-\boldsymbol{\theta}_2\|^2.
	\end{equation*}
	\label{asp:conv}
\end{assumption}

Let $a_k^{(t)}$ be the age of $\boldsymbol{\theta}_k(t)$ relative to $\tilde{\boldsymbol{\theta}}(t)$, which measures the outdatedness of the adopted model before local training.
Then, $a_k^{(t)}=c$ when $\boldsymbol{\theta}_k(t)=\tilde{\boldsymbol{\theta}}(t-c)$, for $c\geq 0$.
%We define the global and local optimum respectively as $F^*=\min F(\boldsymbol{\theta})=F(\boldsymbol{\theta}^*)$ and $F_k^*=\min F_k(\boldsymbol{\theta})=F(\boldsymbol{\theta}^*_k)$, $\forall k$. 
We define 
\begin{equation}
\zeta=\sum_{k\in\mathcal{K}}w_k\|\boldsymbol{\theta}^*-\boldsymbol{\theta}_k^*\|^2\label{eq:zeta_hetero}    
\end{equation}
to quantify the device heterogeneity, where $\boldsymbol{\theta}^*=\argmin F(\boldsymbol{\theta})$ and $\boldsymbol{\theta}_k^*=\argmin F_k(\boldsymbol{\theta})$.
Then, the device model update $\triangle\boldsymbol{\theta}_k(t)$ becomes 
\begin{equation}
	\triangle\boldsymbol{\theta}_k(t)=-\eta\nabla F_k(\tilde{\boldsymbol{\theta}}(t-a_k^{(t)})).\label{eq:grdtAge}
\end{equation}

\begin{assumption}
	A full model is encoded and broadcast in a way that every device can receive it. Consequently, there exists $a_{\lim}>0$ such that $a_k^{(t)}\leq a_{\lim},\forall k,\forall t$. This age limit  $a_{\lim}$ is no larger than the time difference between any two adjacent full-model broadcasts. Furthermore, we assume no quantization error for a full-model broadcast, that is, $Q_0(\boldsymbol{\theta})=\boldsymbol{\theta}, \forall \boldsymbol{\theta}$.
	\label{asp:ageLim}
\end{assumption}

%\subsection{Theoretical Convergence Analysis}
The following is our main theoretical result.
\begin{theorem}
	Under Assumptions \ref{asp:smth}-\ref{asp:ageLim}, with a  stepsize satisfying
	\begin{equation}
		\eta<\frac{\mu}{2L^2\left[2+\hat{\sigma}a_{\lim}+4\hat{\sigma}a_{\lim}^2\left(\sqrt{\zeta}+2\right)\right]},
		\label{ieq:stepSize_ieq}
	\end{equation}
	%$\eta_t=\frac{\beta}{t+\kappa}$ with $\beta>1/\mu$,
%	\begin{align}
%		&\kappa>2\left\{\frac{8L^4\beta^4\hat{\sigma}^2a_{\lim}^4\zeta}{\mu\left(\mu\beta-1\right)^2}+\left(2\kappa_0+5a_{\lim}\right)\left(1+2\sqrt{\zeta\kappa_0/\mu}\right)\right\},\label{ieq:kappaCond}
%	\end{align}
%	\begin{equation*}
%		\kappa_0=L^2\beta^2\left[1+\hat{\sigma}\left(4a_{\lim}^2+3a_{\lim}+1\right)\right]\Big/\left(\mu\beta-1\right),
%	\end{equation*}
where
	\begin{align}
		\hat{\sigma}=\sqrt{\sigma_2}+(\sigma_2+1)\max(\sigma_1,\sqrt{\sigma_1},2),
		\label{eq:sigma}
	\end{align}
	the following result holds:
%	\begin{align}
%		&\mathbb{E}\left[\|\tilde{\boldsymbol{\theta}}(t+1)-\boldsymbol{\theta}^*\|^2\right]\nonumber\\
%		&<\left(1-\eta\mu\varsigma\right)^{\lfloor\frac{t-1}{2a_{\lim}+1}\rfloor+1}\mathbb{E}\left[\|\tilde{\boldsymbol{\theta}}(1)-\boldsymbol{\theta}^*\|^2\right]+\eta^2\epsilon\sum_{i=0}^{t-1}\left(1-\eta\mu\varsigma\right)^i,
%		\label{ieq:optGapMain}
%	\end{align}
	\begin{align}
		&\mathbb{E}\left[\|\tilde{\boldsymbol{\theta}}(t+1)-\boldsymbol{\theta}^*\|^2\right]\nonumber\\
		&\le \left(1-\frac{\eta\mu}{2}\right)^{\lfloor\frac{t-1}{3a_{\lim}+1}\rfloor+1}\mathbb{E}\left[\|\tilde{\boldsymbol{\theta}}(1)-\boldsymbol{\theta}^*\|^2\right]+\frac{2\eta\epsilon}{\mu},
		\label{ieq:optGapMain}
	\end{align}
    where
	\begin{equation}
		\epsilon=L^2\sqrt{\zeta}\left\{\sqrt{\zeta}\left[\hat{\sigma}a_{\lim}\left(2a_{\lim}+1\right)+2\right]+4\hat{\sigma}a_{\lim}^2\right\}.\label{eq:epsilon}
	\end{equation}
	The expectation is taken  w.r.t. the randomness in the  device decoding failures and the quantization of the
    downlink broadcast of $\hat{\boldsymbol{\theta}}(t)$. See Appendix \ref{apx:main_thm} for the proof.
	\label{thm:main}
\end{theorem}
\begin{remark}
	A larger $a_{\lim}$ or a larger $\hat{\sigma}$ requires a smaller  learning rate, as indicated in \eqref{ieq:stepSize_ieq}. This increases  $1-\eta\mu/2$, which slows down the per-iteration contraction  in \eqref{ieq:optGapMain}. 
	Furthermore, $\mathbb{E}\left[\|\tilde{\boldsymbol{\theta}}(t+1)-\boldsymbol{\theta}^*\|^2\right]\rightarrow2\eta\epsilon/\mu$ when $t\rightarrow\infty$. The
    asymptotic error, $\eta\epsilon/\mu$, increases with $\epsilon$, indicating that smaller $a_{\lim}$ and $\hat{\sigma}$ can lead to better learning performance. 
	\label{rmk:main_age}
\end{remark}

%\zheng{We used many times the expression ``work on a model''. While this is understandable in informal discussion, in formal writing I think we should be more specific that "a device works on a model $\theta_k$" means a device uses $\theta_k$ as the current iterate (or solution estimate) before applying gradient descent to compute the next iterate. }
\begin{remark}
With the MTDC scheme, some devices will use outdated models in the local training, which will cause an error in the gradient updates transmitted to the
server. In the analysis, this error term, $\nabla F_k(\tilde{\boldsymbol{\theta}}(t-a_k^{(t)}))-\nabla F_k(\tilde{\boldsymbol{\theta}}(t))$, is treated as extra additive noise on the gradients. The magnitude of this  additive noise 
scales with  $ \Vert \tilde{\boldsymbol{\theta}}(\tau)-\boldsymbol{\theta}^* \Vert, \tau\leq t$. The consequence is a slowdown in the contraction speed. More explicitly, we have the following observation.
\end{remark}
\begin{propi}
\label{propOnly}
Because of the outdated models, the linear convergence factor changes  
from $(1-\eta\mu)^{t}$ nominally to $(1-\eta\mu/2)^{\lfloor\frac{t-1}{3a_{\lim}+1}\rfloor+1}$. 
Since $1-\eta\mu<1-\eta\mu/2<1$ and $\lfloor\frac{t-1}{3a_{\lim}+1}\rfloor+1\leq t$, we have $(1-\eta\mu/2)^{\lfloor\frac{t-1}{3a_{\lim}+1}\rfloor+1}>(1-\eta\mu)^{t}$.
Therefore, this represents a slow-down of the convergence.
\end{propi}
Note that this situation is somewhat akin to stochastic gradient descent with noise
whose conditional second-order moment scales with the iterate, in which case the contraction speed also reduces (see, for example,
\cite{rei2019ext}, \cite[Th. 4.6]{bottou2018ML}, and the NC$^3$T   in \cite{erik2025uni}).
However, we stress that Proposition \ref{propOnly}  is valid for the specific convergence bound 
that we derived in \eqref{ieq:optGapMain}, which in turn holds under Assumptions \ref{asp:smth}-\ref{asp:ageLim}.

The convergence analysis of FedAvg with MTDC-based DL transmission is complex. Our convergence bound is rigorous and captures all phenomena. It makes no assumptions on the statistical distributions of the difference between consecutive global models, and it explains analytically how the model staleness caused by decoding failures impacts convergence. This in turn inspires the age-aware design introduced next.

\section{Federated Learning with Age-Aware Differential Coding and Device Scheduling}
\label{sec:age-aware-design}

The discussion so far has assumed that every device is scheduled for uplink transmission in every iteration. To make more efficient use of uplink resources, we next introduce uplink device scheduling. As shown in Theorem 1, the convergence bound depends on $a_{\lim}$, and the gradient updates obtained from stale models due to decoding failures generally lead to slower convergence. 
Motivated by this result, we propose a dynamic, age-aware version of MTDC and an associated age-aware scheduling policy.
%They are motivated by the fact that decoding failures result in stale local models, and that gradient updates based on such models slow down convergence, which can be seen explicitly
%from  the dependence on $a_{\lim}$ of the right hand side of \eqref{ieq:optGapMain} in Theorem 1. 
%with heterogeneous ``age'' (level of outdatedness). 
%Intuitively, as shown in \eqref{eq:syncFlAggregation}, the misaligned updates from outdated $\triangle\boldsymbol{\theta}_k(t),k\in\mathcal{K},$ contribute deviation to the model evolving course. Moreover, the theoretical analyses reveal that model staleness slows the convergence of learning performance.
%To address this issue, we aim to answer the following questions:
%\begin{itemize}
	%\item how to decide the transmit frequencies of full and different levels of differential model to maintain the balance between model staleness and resource consumption?%
%	\item can we improve model staleness at UL phase?
%\end{itemize}

\subsection{Age-Aware Mixed-Timescale Differential Coding}

Recall that MTDC alleviates the model staleness in the system, as exemplified in Fig. \ref{fig:broadcastTime}b. However, are there alternative ways of deciding $\mathcal{T}_0$, $\mathcal{T}_1$, and $\mathcal{T}_2$ other than the fixed periodic patterns considered in the given example? 
Intuitively, if decoding failures regularly happen, we should broadcast higher-level models more frequently. Otherwise, broadcasting lower-level models is more beneficial since it consumes less communication resources.

As mentioned in Section \ref{sec:mtdc_intro}, with appropriate channel coding, the decoding failure probabilities follow $P^{(t)}_{k,0}\ll P^{(t)}_{k,1}\leq P^{(t)}_{k,2}$.
%\zheng{$A_0, A_1$ and $A_2$ are not defined. }
We define 
\begin{equation}
	\label{eq:A0}
	A_0=\frac{1}{K}\sum_{k\in\mathcal{K}}P^{(t)}_{k,0}\left[a_k^{(t)}+1\right]
\end{equation}
\begin{equation}
	\label{eq:A1}
	A_1=\frac{1}{K}\sum_{k\in\mathcal{K}}\begin{cases}
		a_k^{(t)}+1,&a_k^{(r_{t+1})}>0\\
		P^{(t)}_{k,1}\left[a_k^{(t)}+1\right],&\text{otherwise}
	\end{cases}
\end{equation}
\begin{equation}
	\label{eq:A2}
	A_2=\frac{1}{K}\sum_{k\in\mathcal{K}}\begin{cases}
		a_k^{(t)}+1,&a_k^{(t)}>0\\
		P^{(t)}_{k,2},&\text{otherwise}
	\end{cases}
\end{equation}
Then, at any iteration $t$, we can predict the average device age at the next iteration $t+1$ by computing $A_i$ for $t+1\in\mathcal{T}_i$, $i=0,1,2$. 
%\zheng{Explain why we choose the highest-level first, then to lowest-level broadcast.}
Recall that the required resources for transmitting $\hat{\boldsymbol{\theta}}(t),t\in\mathcal{T}_i$, is decreasing with $i$, as stronger coding protection and more-level data quantization are applied for smaller $i$. 
To balance between improving model staleness and saving communication resources, we decide $t+1\in\mathcal{T}_i$, where $i=\max\{j|A_j\leq \bar{A},j=0,1,2\}$ and $\bar{A}$ is a predetermined age limit. This way, the average device age is expected to be below $\bar{A}$, while for $\hat{\boldsymbol{\theta}}(t)$ with $t \in \mathcal{T}_i$, a larger $i$ is preferred to save communication resources since lower-level differential updates are quantized with fewer bits.
We call this scheme age-aware MTDC (A-MTDC).

\subsection{Age-Aware Device Scheduling}
\label{sec:schAggr}

In the UL transmission phase of FL, device scheduling is typically implemented to reduce the number of communication links. 
Several works propose age-based scheduling designs 1) to guarantee fairness of device participation, which effectively tackles the issue of data heterogeneity \cite{yang2019agebased,hu2024ver,zheng2024aou,liu2026age,wang2025cov}; or 2) to down-weight the contributions of stale updates and therefore improve the learning performance \cite{zeng2025clut}.
In our   system, since the model update $\triangle\boldsymbol{\theta}_k(t)$ based on an outdated model (i.e., $\boldsymbol{\theta}_k(t)=\tilde{\boldsymbol{\theta}}(t-a_k^{(t)})$ with $a_k^{(t)}>0$) may negatively affect the FL performance, we prioritize devices with fresher models, by deciding  $\Pi(t)$ at random  
in every iteration, based on a set of age-aware scheduling probabilities $\{p_k(t)\}_{k=1}^K$.
We select these probabilities as
\begin{equation}
	p_k(t)=\frac{e^{-a_k^{(t)}/a_{\text{max}}}}{\sum_{i\in\mathcal{K}}e^{-a_i^{(t)}/a_{\text{max}}}},\forall k,
	\label{eq:schProb}
\end{equation}
where $a_{\text{max}}=\max_{k\in\mathcal{K}}a_k^{(t)}$.
Consequently,   devices with outdated models will be assigned smaller scheduling probabilities and thus are less likely to participate in the model aggregation.
\begin{remark}
Our proposed device scheduling policy prioritizes devices whose local model updates are computed based on fresher global models. 
 As shown in \cite{liu2026age,wang2025cov}, maintaining fairness in device participation is another important aspect for the learning performance in non-IID (independently and identically distributed) data scenarios. Note that in \cite{liu2026age} and \cite{wang2025cov}, the local updates are computed based on the same global model, while in our framework we need to deal with stale global models at the participating devices.  Finding the optimal balance between participation fairness and information freshness in the scheduling design could be worth exploring in future work.
\end{remark}

We summarize the   operation  of the proposed FL system in \textbf{Algorithm \ref{alg:overall}} (notation in Table \ref{tab:param1}). Detailed steps at the server and  devices are in \textbf{Algorithms \ref{alg:srv_flow}, \ref{alg:dev_flow}}, respectively.
%together with the server/device-side procedures in \textbf{Algorithm \ref{alg:mbc}} and \textbf{\ref{alg:dev_oprt}}.
%\zheng{Did you define the abbreviations used in the algorithm description?}
\begin{table}[t]
	\caption{The notation in the algorithms.} %title of the table
	\centering % centering table
	\normalsize
	\begin{tabular}{|c|c|} % creating eight columns
		\hline
		& \textbf{Definition} \\
		\hline
		%\\[1ex] % inserts single-line
		$\text{Lvl}_t$& $0$: $t\in\mathcal{T}_0$; $1$: $t\in\mathcal{T}_1$; $2$: $t\in\mathcal{T}_2$\\
        \emph{DlModelLvl}$(\cdot)$&the function to decide $t\in\mathcal{T}_0,\mathcal{T}_1,$ or $\mathcal{T}_2$\\
        $\bar{A}$&age limit of A-MTDC\\
		\emph{BroadcastMdl}$(\cdot)$&the function to compute $\hat{\boldsymbol{\theta}}(t)$\\
		\emph{Rct}$(\cdot)$&the function to compute $\boldsymbol{\theta}_k(t)$\\
		\emph{MemUpdateSrv}$(\cdot)$& memory update  at the server\\
		\emph{MemUpdateUsr}$(\cdot)$& memory update  at a device\\
        $p_k(t), k\in\mathcal{K}$&device scheduling probability\\%[1ex] % [1ex] adds vertical space
		\hline % inserts single-line
	\end{tabular}
	\label{tab:param1}
\end{table}

\begin{algorithm}[t!]
	\caption{FL with A-MTDC and Age-aware Device Scheduling}
	\begin{algorithmic}[1]
		\STATE Initialize: $\text{Lvl}_1=0$, $a_k^{(0)}=0, \boldsymbol{\theta}_k(0)=\boldsymbol{0}$, $\forall k$.
		\FOR{$t=1,...,T$}
		\IF{$t>1$}
		\STATE $\text{Lvl}_t\leftarrow$
		\STATE\emph{DlModelLvl}$(\{a_k^{(t)},a_k^{(r_{t+1})}\}_{\forall k},\{P_{k,0}^{(t)},P_{k,1}^{(t)},P_{k,2}^{(t)}\}_{\forall k},\bar{A})$
		\ENDIF
		\STATE $\hat{\boldsymbol{\theta}}(t)\leftarrow$\emph{BroadcastMdl}$\left(\boldsymbol{\theta}(t),\tilde{\boldsymbol{\theta}}(r_{t}),\text{Lvl}_t\right)$
		\STATE \emph{MemUpdateSrv}$(\hat{\boldsymbol{\theta}}(t),\tilde{\boldsymbol{\theta}}(r_{t}),\text{Lvl}_t)$
		\FORALLP{device $k\in\mathcal{K}$}
		\STATE $\boldsymbol{\theta}_k(t)$, $a_k^{(t)}$ $\leftarrow$
		\STATE \emph{Rct}$\left(\hat{\boldsymbol{\theta}}(t),\boldsymbol{\theta}_k(t-1),\tilde{\boldsymbol{\theta}}(r_{t}),a_k^{(t-1)},\text{Lvl}_t\right)$
		\STATE\emph{MemUpdateUsr}$(\boldsymbol{\theta}_k(t),a_k^{(t)})$
		\ENDFAP
		\STATE The server computes \eqref{eq:schProb} to determine $\Pi(t)$.
		\FORALLP{device $k\in\Pi(t)$}
		\STATE Local training with $\boldsymbol{\theta}_k(t)$, obtain $\triangle\boldsymbol{\theta}_k(t)$ and transmit it to the server.
		\ENDFAP
		\STATE The server computes $\sum_{k\in\Pi(t)}w_k(t)\triangle\boldsymbol{\theta}_k(t)$ and renews the model by \eqref{eq:update_MTDC}. 
		\ENDFOR  		
	\end{algorithmic}
	\label{alg:overall}
\end{algorithm}
\begin{algorithm}[t!]
	\caption{Server Operations}
	\begin{algorithmic}[1]
	\STATE Lvl$_t=$\emph{DlModelLvl}$\left(\{a_k^{(t)},a_k^{(r_{t+1})}\}_{\forall k},\{P_{k,0}^{(t)},P_{k,1}^{(t)},P_{k,2}^{(t)}\}_{\forall k},\bar{A}\right)$:
	\STATE Compute $\{A_i\}_{i=0}^2$ in \eqref{eq:A0}-\eqref{eq:A2}.
	\STATE $\text{Lvl}_t=\max\{j|A_j\leq \bar{A},j=0,1,2\}$.
	\STATE
	\STATE $\hat{\boldsymbol{\theta}}(t)=\text{\emph{BroadcastMdl}}\left(\boldsymbol{\theta}(t), \tilde{\boldsymbol{\theta}}(r_{t}), \text{Lvl}_t\right)$:
	\IF{$\text{Lvl}_t=0$}
	\STATE $\hat{\boldsymbol{\theta}}(t)\leftarrow Q_0\left(\boldsymbol{\theta}(t)\right)$.
	\ELSE
	\STATE $i\leftarrow\text{Lvl}_t$, $\hat{\boldsymbol{\theta}}(t)\leftarrow Q_i\left(\boldsymbol{\theta}(t)-\tilde{\boldsymbol{\theta}}(r_{t})\right)$.
	\ENDIF
	\STATE
	\STATE \emph{MemUpdateSrv}$(\hat{\boldsymbol{\theta}}(t),\tilde{\boldsymbol{\theta}}(r_{t}),\text{Lvl}_t)$:
	\IF{$\text{Lvl}_t=0$}
	\STATE Save $\tilde{\boldsymbol{\theta}}(t)=\hat{\boldsymbol{\theta}}(t)$ to the memory.
	\ELSE
	\STATE Save $\tilde{\boldsymbol{\theta}}(t)=\hat{\boldsymbol{\theta}}(t)+\tilde{\boldsymbol{\theta}}(r_{t})$ to the memory.
	\ENDIF
    \end{algorithmic}
    \label{alg:srv_flow}
\end{algorithm}
\begin{algorithm}[t!]
	\caption{Device Operations}
	\begin{algorithmic}[1]
	\STATE $\boldsymbol{\theta}_k(t), a_k^{(t)}=\text{\emph{Rct}}\left(\hat{\boldsymbol{\theta}}(t),\boldsymbol{\theta}_k(t-1),\tilde{\boldsymbol{\theta}}(r_{t}),a_k^{(t-1)},\text{Lvl}_t\right)$:
	\IF{fail to decode $\hat{\boldsymbol{\theta}}(t)$}
	\STATE $\boldsymbol{\theta}_k(t)=\boldsymbol{\theta}_k(t-1)$, $a_k^{(t)}=a_k^{(t-1)}+1$.
	\ELSIF{$\text{Lvl}_t=0$}
	\STATE $\boldsymbol{\theta}_k(t)=\hat{\boldsymbol{\theta}}(t)$, $a_k^{(t)}=0$.
	\ELSIF{$\tilde{\boldsymbol{\theta}}(r_{t})$ available}
	\STATE $\boldsymbol{\theta}_k(t)=\hat{\boldsymbol{\theta}}(t)+\tilde{\boldsymbol{\theta}}(r_{t})$, $a_k^{(t)}=0$.
	\ELSE
	\STATE $\boldsymbol{\theta}_k(t)=\boldsymbol{\theta}_k(t-1)$, $a_k^{(t)}=a_k^{(t-1)}+1$.
	\ENDIF
	\STATE
	\STATE \emph{MemUpdateUsr}$(\boldsymbol{\theta}_k(t),a_k^{(t)})$:
	\STATE Save $\boldsymbol{\theta}_k(t)$ to the memory.
	\IF{$a_k^{(t)}=0$}
	\STATE Save $\tilde{\boldsymbol{\theta}}(t)=\boldsymbol{\theta}_k(t)$ to the memory.
	\ENDIF
    \end{algorithmic}
    \label{alg:dev_flow}
\end{algorithm}

\section{Simulations}
\label{sec:simulation_results}
We train two convolutional neural networks parameterized by $\boldsymbol{\theta}\in\mathbb{R}^{21840}$ and $\boldsymbol{\theta}\in\mathbb{R}^{62006}$, with MNIST \cite{lecun-mnisthandwrittendigit-2010} and CIFAR-10 \cite{Krizhevsky09learningmultiple} datasets,  respectively. There are $K=20$ devices in the system. The training data are allocated to each device in a non-IID fashion. Each device contains training data of up to $6$ different classes/labels.
To evaluate the learning performance, at every iteration, the global model is tested on the testing datasets. The test accuracy (defined as the percentage of correct classification instances) is then measured.

\vspace{-0.2cm}
\subsection{Gain of Mixed-Timescale Differential Coding}
\label{sec:firstSim}
We evaluate the performance of the following schemes:
\begin{itemize}
	\item '\emph{AllFull}': baseline method, always broadcasting a full model in every iteration, i.e., $t\in\mathcal{T}_0,\forall t$.
	\item '\emph{DiC-$\rho$}': state-of-the-art method \cite{amiri2022Con}, broadcasting a full model every $\rho$ iterations
    and a differential update at all other iterations, i.e.,  
	\begin{equation*}
		t\in\begin{cases}
			\mathcal{T}_0, &t=1+\rho n, n=0,1,...\\
			\mathcal{T}_2, &\text{otherwise}	
		\end{cases}
	\end{equation*}
    To ensure a fair performance comparison, we periodically allocate the full model to DiC \cite{amiri2022Con}.
	\item '\emph{MTDC}-$(\rho_1,\rho_2)$': proposed method, broadcasting either a full model,
    with fixed period $\rho_1$; or a   first-level  differential update, with period  $\rho_2$;
    or a second-level differential update, in all other iterations: 
	\begin{equation*}
		t\in\begin{cases}
			\mathcal{T}_0, &t=1+\rho_1n, n=0,1,...\\
			\mathcal{T}_1, &t=1+\rho_2n,t\neq1+\rho_1n, n=0,1,...\\
			\mathcal{T}_2, &\text{otherwise.}	
		\end{cases}
	\end{equation*} 
\end{itemize}
The decoding failure probabilities are fixed over time and across devices, i.e., $\{P_{k,i}^{(t)}\}_{i=0}^{2}=[0.001,0.2,0.25], \forall k, \forall t$.
We schedule the devices uniformly at random, with scheduling ratios $|\Pi(t)|/K=0.1$, and $0.25$, respectively for the datasets MNIST and CIFAR-10. 
For fair comparison between the different methods (full model broadcast, DiC, and MTDC), we keep the \emph{time-average of the DL transmission bit rate} approximately the same (up to rounding effects). The exact number of bits in a given iteration may vary between different methods. Following this guideline, the quantization levels ($\nu_i$) are chosen as:
\begin{itemize}
\item '\emph{AllFull}': $\nu_0=31$
\item '\emph{DiC-$\rho$}': $\nu_0=255$ and $\nu_2=15$
\item '\emph{MTDC}-$(\rho_1,\rho_2)$': $\nu_0=255$, $\nu_1=127$, and $\nu_2=7$
\end{itemize}
The assignment of $\{ \nu_i \}$ is consistent with the ordering discussed in Remark \ref{rmk:rdmqnt}. For the case of MNIST, the average bit rate of each method is illustrated in Fig. \ref{fig:schComp_b}. With this normalization, all methods consume approximately the same amount of communication resources over time. Minor discrepancies arise due to rounding effects in the selection of quantization levels. In general, the average bit rate increases as the frequency of higher-resolution broadcasts increases.
%\begin{remark}
%\textcolor{blue}{
%Compared to AllFull, the two-level MTDC requires additionally $\lceil\log_2 (t-r_t)\rceil$ and $2$ bits, respectively, for the reference model  timestamp recording, and to convey  the broadcast model type in each transmission block. For instance, MTDC-(10,5) requires less than a $1\%$ bit rate increase for control signaling, which is negligible compared to the $18.3\%$ bit rate saving over AllFull, as indicated in Table \ref{tab:bitTab}. Since the model size generally exceeds the maximum packet size, with a similar time-average bit rate and the same payload size, MTDC would not require a higher channel coding redundancy or extra signaling overhead.
As discussed earlier, our MTDC scheme requires extra memory usage. As an example, the $62006$-parameter convolutional neural network model (for the  CIFAR-10 dataset), with $32$-bit precision, requires less than $1$ MB memory. This is relatively little compared to the storage capacity of modern edge devices.
%}
%\end{remark}
% \begin{table}[h]
% 	\caption{\textcolor{blue}{The extra bit rate  relative to \emph{MTDC}-(10,5), based on Fig. \ref{fig:schComp_b}.}} %title of the table
% 	\centering % centering table
% 	\normalsize
% 	\begin{tabular}{ccccc} % creating eight columns
% 		\hline
% 		\emph{AllFull}&\emph{DiC}-5&\emph{DiC}-7&\emph{DiC}-10&\emph{MTDC}-(8,4)  \\
% 		\hline
% 		$18.3\%$&$15.5\%$&$12.5\%$&$9.3\%$&$4.5\%$\\
% 		\hline % inserts single-line
% 	\end{tabular}
% 	\label{tab:bitTab}
% \end{table} 

Figs.~\ref{fig:schComp_lern} and \ref{fig:schCompCifar} show the test accuracy comparison between different schemes, for the datasets MNIST and CIFAR-10, respectively. The \emph{AllFull} method performs the worst, since it transmits the full model at every iteration and therefore relies on low-resolution quantization. In contrast, the differential coding schemes -- both the state-of-the-art DiC and the proposed MTDC -- mitigate this limitation by transmitting updates with a smaller dynamic range that can be more aggressively compressed.
Moreover, for the differential coding schemes, more frequent higher-level broadcasts generally lead to improved test accuracy (e.g., in Fig. \ref{fig:schComp_lern}, \emph{DiC}-5 outperforms \emph{DiC}-10, and \emph{MTDC}-(8,4) outperforms \emph{MTDC}-(10,5)), albeit at the cost of slightly higher average bit rates. Finally, and most importantly, MTDC consistently achieves strong learning performance with lower communication overhead than the state-of-the-art methods, owing to its increased resilience to decoding failures.\footnote{Since the conclusions from the experiments with  CIFAR-10 are consistent with those from MNIST, we only show the results for MNIST for the remaining experiments.}

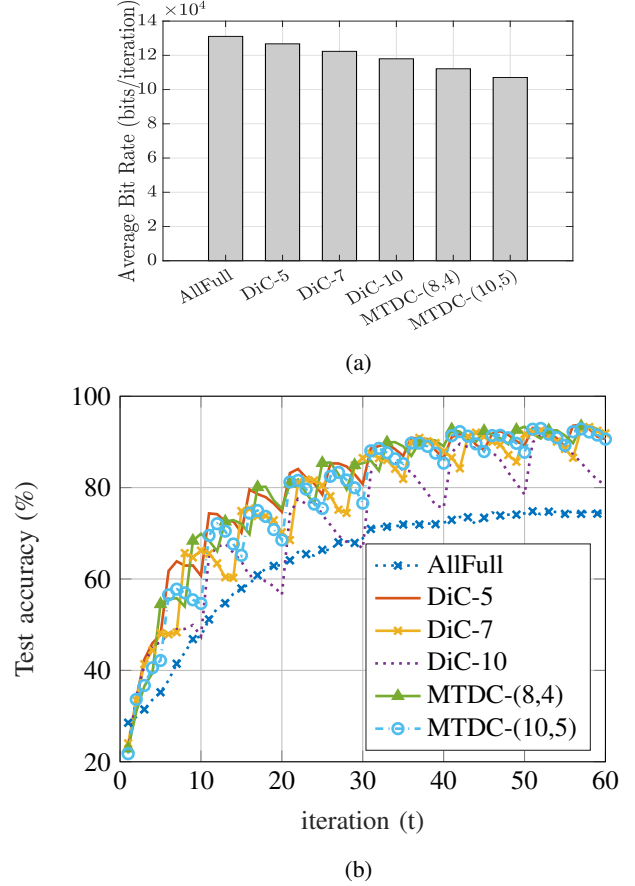
\begin{figure}[t!]
	\centering
    	\begin{subfigure}[c]{.5\textwidth}
		\centering
		\includegraphics[width=7cm, height=4.5cm]{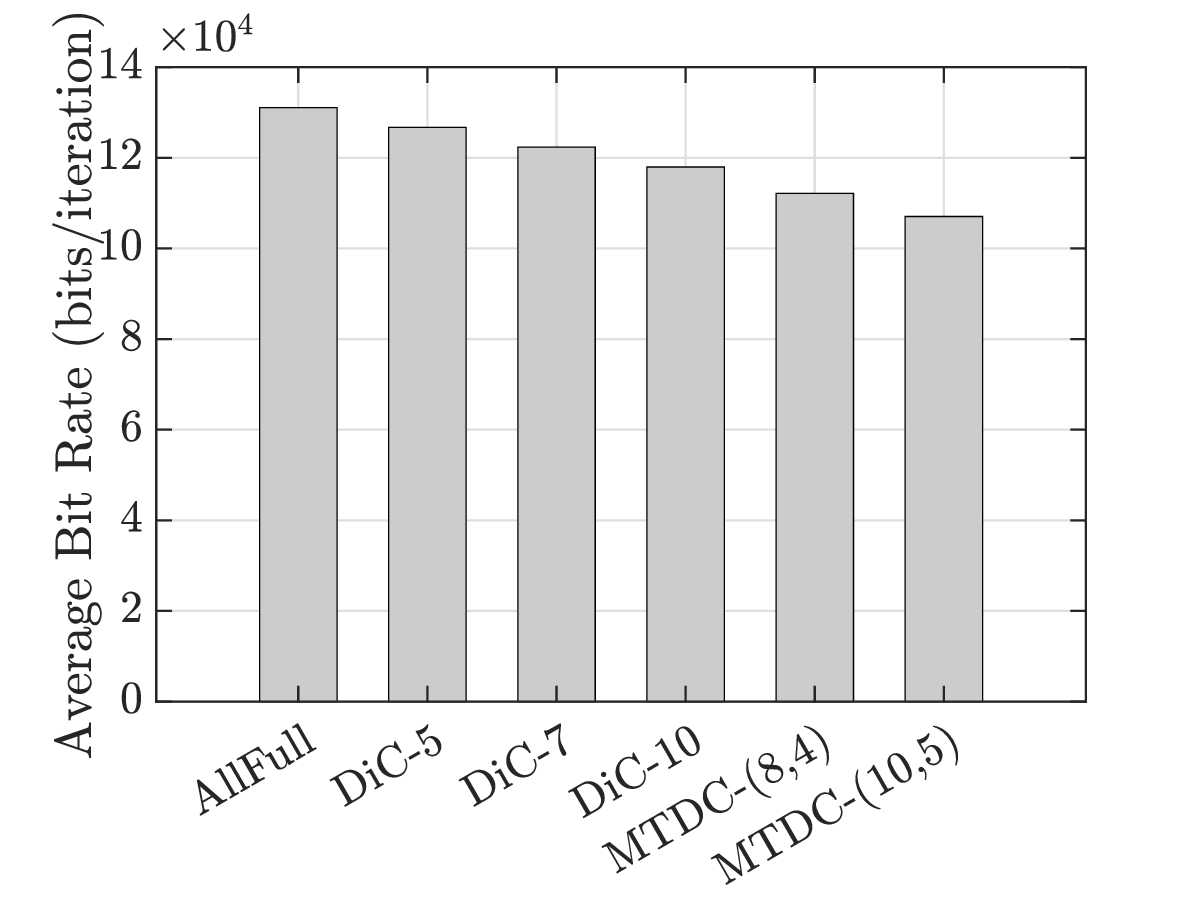}
        \caption{}
        \label{fig:schComp_b}
	\end{subfigure}
	\hfill
    \begin{subfigure}[c]{.5\textwidth}
		\hspace{-.25cm}
\definecolor{mycolor1}{rgb}{0.00000,0.44700,0.74100}%
\definecolor{mycolor2}{rgb}{0.85000,0.32500,0.09800}%
\definecolor{mycolor3}{rgb}{0.92900,0.69400,0.12500}%
\definecolor{mycolor4}{rgb}{0.49400,0.18400,0.55600}%
\definecolor{mycolor5}{rgb}{0.46600,0.67400,0.18800}%
\definecolor{mycolor6}{rgb}{0.30100,0.74500,0.93300}%
\begin{tikzpicture}

\begin{axis}[%
width=2.528in,
height=1.904in,
at={(0.758in,0.651in)},
scale only axis,
xmin=0,
xmax=60,
xlabel style={font=\color{white!15!black}},
xlabel={iteration (t)},
ymin=20,
ymax=100,
ylabel style={font=\color{white!15!black}},
ylabel={Test accuracy (\%)},
axis background/.style={fill=white},
xmajorgrids,
ymajorgrids,
legend style={at={(0.97,0.03)}, anchor=south east, legend cell align=left, align=left, draw=white!15!black}
]
\addplot [color=mycolor1, dotted, line width=1.0pt, mark=x, mark options={solid, mycolor1}, mark repeat=2]
  table[row sep=crcr]{%
1	28.5433333333333\\
2	29.8\\
3	31.4533333333333\\
4	33.5766666666667\\
5	35.2066666666667\\
6	37.84\\
7	41.4633333333333\\
8	43.8833333333333\\
9	46.8366666666667\\
10	47.5\\
11	51.1533333333333\\
12	53.53\\
13	54.6966666666667\\
14	56.77\\
15	57.9166666666667\\
16	58.9333333333333\\
17	60.8266666666667\\
18	62.24\\
19	62.8566666666667\\
20	63.0566666666667\\
21	64.0833333333333\\
22	65.9\\
23	65.4466666666667\\
24	65.1333333333333\\
25	66.37\\
26	66.36\\
27	68.02\\
28	68.31\\
29	67.9\\
30	68.8866666666667\\
31	70.9666666666667\\
32	71.6266666666667\\
33	71.4433333333333\\
34	71.5033333333333\\
35	71.9266666666667\\
36	72.0266666666667\\
37	71.92\\
38	72.1566666666667\\
39	71.99\\
40	72.34\\
41	72.9266666666667\\
42	73.5766666666667\\
43	73.51\\
44	71.93\\
45	73.41\\
46	74.5566666666667\\
47	73.89\\
48	74.08\\
49	74.09\\
50	74.3666666666667\\
51	74.8066666666667\\
52	74.17\\
53	74.7333333333333\\
54	74.6633333333333\\
55	74.25\\
56	74.8733333333333\\
57	74.25\\
58	74.9833333333333\\
59	74.3166666666667\\
60	74.61\\
};
\addlegendentry{AllFull}

\addplot [color=mycolor2, solid, line width=1.0pt]
  table[row sep=crcr]{%
1	22.2766666666667\\
2	33.6766666666667\\
3	42.4033333333333\\
4	45.9666666666667\\
5	47.7933333333333\\
6	61.7866666666667\\
7	63.8733333333333\\
8	62.93\\
9	62.9666666666667\\
10	60.7133333333333\\
11	74.3566666666667\\
12	74.1833333333333\\
13	72.57\\
14	72.64\\
15	70.1566666666667\\
16	79.5966666666667\\
17	78.5866666666667\\
18	77.9133333333333\\
19	76.5033333333333\\
20	74.69\\
21	83.1866666666667\\
22	84.0733333333333\\
23	82.1633333333333\\
24	80.8866666666667\\
25	78.46\\
26	85.2866666666667\\
27	85.28\\
28	84.6233333333333\\
29	82.9133333333333\\
30	80.7566666666667\\
31	88.5733333333333\\
32	89.0933333333333\\
33	89.5733333333333\\
34	88.3966666666667\\
35	86.93\\
36	90.91\\
37	90.69\\
38	90.61\\
39	88.1466666666667\\
40	86.3833333333333\\
41	91.0666666666667\\
42	91.71\\
43	91.43\\
44	89.3866666666667\\
45	88.19\\
46	91.81\\
47	92.3566666666667\\
48	91.5133333333333\\
49	90.14\\
50	89.24\\
51	93.0566666666667\\
52	92.8366666666667\\
53	91.9\\
54	90.8\\
55	88.9733333333333\\
56	93.55\\
57	93.7433333333333\\
58	92.81\\
59	91.1366666666667\\
60	90.6033333333333\\
};
\addlegendentry{DiC-5}

\addplot [color=mycolor3, line width=1.0pt, mark=x, mark options={solid, mycolor3}]
  table[row sep=crcr]{%
1	24.0366666666667\\
2	33.4766666666667\\
3	41.3766666666667\\
4	44.45\\
5	48.1933333333333\\
6	47.8866666666667\\
7	48.3266666666667\\
8	65.67\\
9	64.7233333333333\\
10	66.2033333333333\\
11	65.5633333333333\\
12	63.43\\
13	60.3966666666667\\
14	60.2566666666667\\
15	74.8766666666667\\
16	75.3866666666667\\
17	73.7633333333333\\
18	73.8\\
19	72.8666666666667\\
20	70.2266666666667\\
21	68.51\\
22	81.0866666666667\\
23	81.7833333333333\\
24	81.5666666666667\\
25	80.5033333333333\\
26	78.1166666666667\\
27	75.1566666666667\\
28	74.4533333333333\\
29	84.8466666666667\\
30	86.4866666666667\\
31	87.4533333333333\\
32	86.4466666666667\\
33	86.0366666666667\\
34	84.1133333333333\\
35	81.8966666666667\\
36	89.89\\
37	90.8633333333333\\
38	90.1466666666667\\
39	89.7533333333333\\
40	87.9166666666667\\
41	86.5533333333333\\
42	84.2233333333333\\
43	91.21\\
44	92.02\\
45	91.1533333333333\\
46	90.0466666666667\\
47	89.3633333333333\\
48	87.13\\
49	85.67\\
50	91.7066666666667\\
51	92.1233333333333\\
52	92.4766666666667\\
53	91.3366666666667\\
54	90.0733333333333\\
55	88.4766666666666\\
56	86.5366666666667\\
57	92.8666666666667\\
58	93.2633333333333\\
59	92.4466666666667\\
60	91.8733333333333\\
};
\addlegendentry{DiC-7}

\addplot [color=mycolor4, dotted, line width=1.0pt]
  table[row sep=crcr]{%
1	22.6966666666667\\
2	34.2566666666667\\
3	41.1733333333333\\
4	45.3766666666667\\
5	46.1466666666667\\
6	48.08\\
7	48.99\\
8	49.07\\
9	49.92\\
10	47.4066666666667\\
11	65.1066666666667\\
12	66.97\\
13	67.76\\
14	65.75\\
15	63.89\\
16	60.6166666666667\\
17	60.84\\
18	59.4333333333333\\
19	58.4266666666667\\
20	56.54\\
21	74.2333333333333\\
22	77.5366666666667\\
23	76.4266666666667\\
24	77.11\\
25	74.0433333333333\\
26	71.8633333333333\\
27	69.4933333333333\\
28	68.13\\
29	68.07\\
30	66.67\\
31	83.2866666666667\\
32	85.97\\
33	85.9533333333333\\
34	85.42\\
35	83.7833333333333\\
36	83.54\\
37	81.0633333333333\\
38	78.7866666666666\\
39	76.1466666666667\\
40	75.5166666666667\\
41	87.68\\
42	89.5933333333333\\
43	88.8766666666667\\
44	88.6566666666667\\
45	89.0466666666667\\
46	87.7733333333333\\
47	86.0366666666667\\
48	83.3033333333333\\
49	80.9233333333333\\
50	78.4866666666667\\
51	89.8833333333333\\
52	91.57\\
53	91.0333333333333\\
54	90.0333333333333\\
55	87.9733333333333\\
56	87.35\\
57	85.4266666666667\\
58	83.5233333333333\\
59	81.6433333333333\\
60	80.6866666666667\\
};
\addlegendentry{DiC-10}

\addplot [color=mycolor5, solid, line width=1.0pt, mark=triangle*, mark repeat=4]
  table[row sep=crcr]{%
1	22.8966666666667\\
2	30.97\\
3	36.5166666666667\\
4	39.41\\
5	54.4966666666667\\
6	55.6\\
7	55.81\\
8	53.9933333333333\\
9	68.37\\
10	69.8933333333333\\
11	68.41\\
12	66.09\\
13	72.5633333333333\\
14	72.8833333333333\\
15	72.09\\
16	69.8166666666667\\
17	80.1033333333333\\
18	80.1933333333333\\
19	77.1766666666667\\
20	75.4266666666667\\
21	80.8733333333333\\
22	81.7766666666667\\
23	78.72\\
24	76.47\\
25	85.3833333333333\\
26	85.4966666666667\\
27	82.92\\
28	79.6266666666667\\
29	84.9166666666667\\
30	85.3533333333333\\
31	85.9233333333333\\
32	83.69\\
33	89.8833333333333\\
34	89.95\\
35	89.0366666666667\\
36	86.93\\
37	89.83\\
38	90.61\\
39	90.2933333333333\\
40	89.0366666666667\\
41	92.7333333333333\\
42	92.23\\
43	90.71\\
44	90.0466666666667\\
45	92.3866666666667\\
46	91.6033333333333\\
47	90.7933333333333\\
48	89.5033333333333\\
49	92.6433333333333\\
50	93.34\\
51	91.76\\
52	90.73\\
53	92.5333333333333\\
54	92.5366666666667\\
55	91.3766666666667\\
56	89.7366666666667\\
57	93.4433333333333\\
58	93.53\\
59	92.82\\
60	91.8833333333333\\
};
\addlegendentry{MTDC-(8,4)}

\addplot [color=mycolor6, dashdotted, line width=1.0pt, mark=o, mark options={solid, mycolor6}]
  table[row sep=crcr]{%
1	21.7666666666667\\
2	33.6366666666667\\
3	36.6966666666667\\
4	40.58\\
5	42.1666666666667\\
6	56.5266666666667\\
7	57.7766666666667\\
8	56.9666666666667\\
9	55.4466666666667\\
10	54.6433333333333\\
11	69.5233333333333\\
12	72.0966666666667\\
13	70.4666666666667\\
14	67.6566666666667\\
15	65.1766666666667\\
16	74.4033333333333\\
17	74.9833333333333\\
18	73.7\\
19	70.8166666666667\\
20	68.5333333333333\\
21	81.2366666666667\\
22	81.5866666666667\\
23	79.6633333333333\\
24	76.39\\
25	75.43\\
26	82.4733333333333\\
27	83.3366666666667\\
28	81.66\\
29	79.96\\
30	76.5933333333333\\
31	88.0533333333333\\
32	88.59\\
33	87.63\\
34	86.2966666666667\\
35	85.23\\
36	89.7266666666667\\
37	89.73\\
38	89.0266666666667\\
39	87.5466666666667\\
40	85.35\\
41	91.36\\
42	92.25\\
43	91.26\\
44	89.5766666666667\\
45	87.89\\
46	91.39\\
47	91.45\\
48	90.9\\
49	89.69\\
50	87.7333333333333\\
51	92.7766666666667\\
52	92.96\\
53	91.6866666666667\\
54	90.9066666666667\\
55	89.2333333333333\\
56	92.2266666666667\\
57	92.84\\
58	92.17\\
59	91.5\\
60	90.61\\
};
\addlegendentry{MTDC-(10,5)}

\end{axis}
\end{tikzpicture}%
        \caption{}
        \label{fig:schComp_lern}
	\end{subfigure}
	\caption{Comparison of average communication resource consumption and test accuracy for different schemes (MNIST).}
	\label{fig:schComp}
\end{figure}
\begin{figure}[t!]
	\centering
    \begin{subfigure}[c]{.5\textwidth} %0.6 size to letter size
		\hspace{-.25cm}
        \input{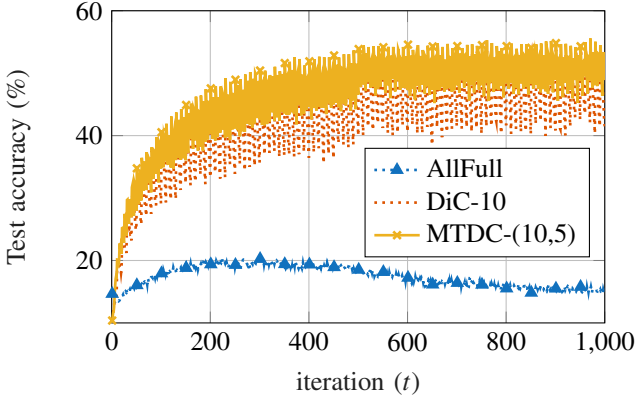}
        %\caption{\textcolor{blue}{CIFAR-10}}
        %\label{fig:schComp_lernCifar}
	\end{subfigure}
	\caption{Comparison of test accuracy for different schemes (CIFAR-10).}
	\label{fig:schCompCifar}
\end{figure}

\subsection{A-MTDC versus MTDC}
To demonstrate the effectiveness of A-MTDC, we simulate the following two vanilla MTDC methods for comparison: \emph{MTDC}-$(10,5)$ and \emph{MTDC}-$(15,5)$. 
Two cases of device decoding failure probabilities are considered: $\{P_{k,i}^{(t)}\}_{i=0}^2=[0.0005, 0.05,0,1]$ and $\{P_{k,i}^{(t)}\}_{i=0}^2=[0.0005, 0.1,0,3]$. The device scheduling ratios $|\Pi(t)|/K$ are $0.5$ or $0.1$. We set the precision of the random quantizer to $\{\nu_i\}_{i=0}^2=[127,63,15]$.

When decoding failures are relatively rare, as in Fig. \ref{fig:dynFixComp}, A-MTDC (with $\bar{A}=2$) and the vanilla MTDC methods perform similarly in test accuracy for both scenarios $|\Pi(t)|/K=0.5$ (curves with legend '$0.5$:') and $|\Pi(t)|/K=0.1$ (legend '$0.1$:'). As A-MTDC keeps track of the device age along the learning process and dynamically chooses an appropriate model type to broadcast, the test accuracy improves more smoothly over time compared to the vanilla MTDC methods. All methods consume a similar amount of communication resources, according to the calculation from Remark \ref{rmk:qnt_bits}.
%\footnote{The calculation of resource consumption is based on Remark \ref{rmk:qnt_bits}.\label{ftQun}}
%\zheng{Specify how you measure ``communication resources''.}

With more frequent  decoding failures, as in Fig. \ref{fig:dynFixComp1}, the vanilla MTDC methods face regular test accuracy drops while A-MTDC does not, thanks to its quick reaction to model staleness. The price to pay for this is additional communication resources ($5\%$ more than vanilla MTDC in this case).
%$^{\ref{ftQun}}$
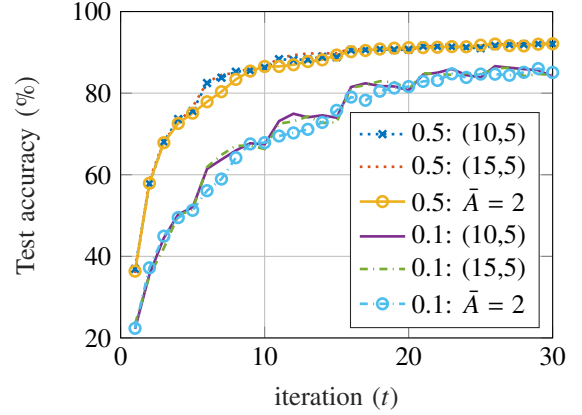
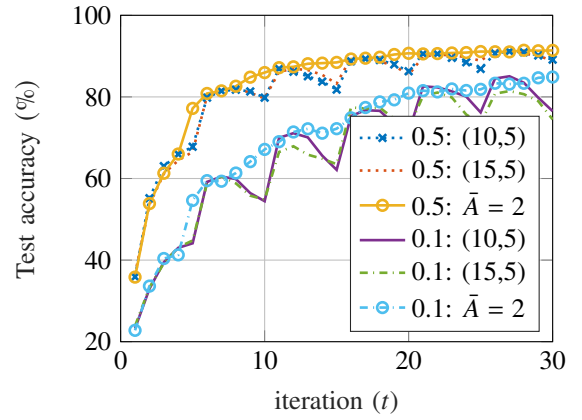
\begin{figure}[t]
	\centering
	\begin{subfigure}[c]{.5\textwidth}
		\centering
\definecolor{mycolor1}{rgb}{0.00000,0.44700,0.74100}%
\definecolor{mycolor2}{rgb}{0.85000,0.32500,0.09800}%
\definecolor{mycolor3}{rgb}{0.92900,0.69400,0.12500}%
\definecolor{mycolor4}{rgb}{0.49400,0.18400,0.55600}%
\definecolor{mycolor5}{rgb}{0.46600,0.67400,0.18800}%
\definecolor{mycolor6}{rgb}{0.30100,0.74500,0.93300}%
\begin{tikzpicture}

\begin{axis}[%
width=2.25in,
height=1.7in,
at={(0.758in,0.651in)},
scale only axis,
xmin=0,
xmax=30,
xlabel style={font=\color{white!15!black}},
xlabel={iteration ($t$)},
ymin=20,
ymax=100,
ylabel style={font=\color{white!15!black}},
ylabel={Test accuracy $(\%)$},
axis background/.style={fill=white},
xmajorgrids,
ymajorgrids,
legend style={at={(0.97,0.03)}, anchor=south east, legend cell align=left, align=left, draw=white!15!black}
]
\addplot [color=mycolor1, dotted, line width=1.0pt, mark=x, mark options={solid, mycolor1}]
  table[row sep=crcr]{%
1	36.92\\
2	57.82\\
3	68.09\\
4	73.645\\
5	75.615\\
6	82.48\\
7	83.795\\
8	85.245\\
9	85.54\\
10	86.265\\
11	88.32\\
12	88.125\\
13	88.28\\
14	88.92\\
15	88.81\\
16	90.585\\
17	90.65\\
18	90.785\\
19	90.62\\
20	90.665\\
21	91.535\\
22	91.455\\
23	91.34\\
24	91.195\\
25	91.015\\
26	91.58\\
27	91.93\\
28	91.97\\
29	92.075\\
30	92.05\\
};
\addlegendentry{0.5: (10,5)}

\addplot [color=mycolor2, dotted, line width=1.0pt]
  table[row sep=crcr]{%
1	35.885\\
2	58.625\\
3	68.32\\
4	72.7\\
5	76.62\\
6	82.63\\
7	84.495\\
8	84.985\\
9	86.1\\
10	86.245\\
11	88.62\\
12	89.315\\
13	89.73\\
14	89.735\\
15	89.805\\
16	90.93\\
17	90.59\\
18	90.81\\
19	90.515\\
20	90.61\\
21	91.355\\
22	91.655\\
23	91.93\\
24	91.67\\
25	91.44\\
26	91.945\\
27	92.065\\
28	92.205\\
29	92.255\\
30	91.96\\
};
\addlegendentry{0.5: (15,5)}

\addplot [color=mycolor3, line width=1.0pt, mark=o, mark options={solid, mycolor3}]
  table[row sep=crcr]{%
1	36.37\\
2	57.915\\
3	67.855\\
4	72.585\\
5	75.095\\
6	77.895\\
7	80.36\\
8	83.38\\
9	85.405\\
10	86.545\\
11	86.555\\
12	86.965\\
13	87.75\\
14	88.165\\
15	89.34\\
16	90.16\\
17	90.455\\
18	90.82\\
19	91.02\\
20	91.15\\
21	91.165\\
22	91.255\\
23	91.41\\
24	91.365\\
25	91.76\\
26	92.03\\
27	91.645\\
28	91.57\\
29	92\\
30	92.105\\
};
\addlegendentry{0.5: $\bar{A}=2$}

\addplot [color=mycolor4, line width=1.0pt]
  table[row sep=crcr]{%
1	22.235\\
2	35.595\\
3	44.5\\
4	50.275\\
5	52.01\\
6	61.425\\
7	63.72\\
8	65.95\\
9	67.705\\
10	67.465\\
11	73.155\\
12	74.96\\
13	74.08\\
14	74.57\\
15	73.975\\
16	81.5\\
17	82.42\\
18	81.735\\
19	81.65\\
20	80.71\\
21	84.545\\
22	85.02\\
23	85.93\\
24	84.435\\
25	83.87\\
26	86.61\\
27	86.24\\
28	85.89\\
29	85.63\\
30	84.03\\
};
\addlegendentry{0.1: (10,5)}

\addplot [color=mycolor5, dashdotted, line width=1.0pt]
  table[row sep=crcr]{%
1	23.43\\
2	35.31\\
3	42.01\\
4	50.06\\
5	51.445\\
6	61.875\\
7	64.87\\
8	67.1\\
9	67.14\\
10	66.3\\
11	72.535\\
12	73.07\\
13	74.3\\
14	72.7\\
15	72.85\\
16	81.355\\
17	81.945\\
18	82.925\\
19	82.69\\
20	80.905\\
21	84.89\\
22	84.65\\
23	84.49\\
24	84.9\\
25	84.565\\
26	86.075\\
27	86.245\\
28	84.175\\
29	84.695\\
30	84.455\\
};
\addlegendentry{0.1: (15,5)}

\addplot [color=mycolor6, dashdotted, line width=1.0pt, mark=o, mark options={solid, mycolor6}]
  table[row sep=crcr]{%
1	22.36\\
2	37.19\\
3	44.945\\
4	49.5\\
5	51.28\\
6	56.095\\
7	58.945\\
8	64.175\\
9	67.56\\
10	67.87\\
11	69.5\\
12	70.25\\
13	71.135\\
14	72.835\\
15	75.78\\
16	79.04\\
17	78.25\\
18	80.48\\
19	81.265\\
20	81.61\\
21	82.89\\
22	83.12\\
23	84.73\\
24	83.85\\
25	84.675\\
26	84.655\\
27	84.385\\
28	85.1\\
29	86.025\\
30	85.075\\
};
\addlegendentry{0.1: $\bar{A}=2$}

\end{axis}
\end{tikzpicture}%
		\caption{$\{P_{k,i}^{(t)}\}_{i=0}^2=[0.0005, 0.05,0,1]$.}
		\label{fig:dynFixComp}
	\end{subfigure}
	\hfill
	\begin{subfigure}[c]{.5\textwidth}
		\centering
\definecolor{mycolor1}{rgb}{0.00000,0.44700,0.74100}%
\definecolor{mycolor2}{rgb}{0.85000,0.32500,0.09800}%
\definecolor{mycolor3}{rgb}{0.92900,0.69400,0.12500}%
\definecolor{mycolor4}{rgb}{0.49400,0.18400,0.55600}%
\definecolor{mycolor5}{rgb}{0.46600,0.67400,0.18800}%
\definecolor{mycolor6}{rgb}{0.30100,0.74500,0.93300}%
\begin{tikzpicture}

\begin{axis}[%
width=2.25in,
height=1.7in,
at={(0.758in,0.651in)},
scale only axis,
xmin=0,
xmax=30,
xlabel style={font=\color{white!15!black}},
xlabel={iteration ($t$)},
ymin=20,
ymax=100,
ylabel style={font=\color{white!15!black}},
ylabel={Test accuracy $(\%)$},
axis background/.style={fill=white},
xmajorgrids,
ymajorgrids,
legend style={at={(0.97,0.03)}, anchor=south east, legend cell align=left, align=left, draw=white!15!black}
]
\addplot [color=mycolor1, dotted, line width=1.0pt, mark=x, mark options={solid, mycolor1}]
  table[row sep=crcr]{%
1	35.895\\
2	55.165\\
3	63.025\\
4	65.955\\
5	67.79\\
6	79.825\\
7	81.445\\
8	81.77\\
9	81.335\\
10	79.815\\
11	86.89\\
12	86.175\\
13	85.065\\
14	83.75\\
15	81.795\\
16	88.74\\
17	89.315\\
18	89.035\\
19	87.965\\
20	86.29\\
21	90.57\\
22	90.555\\
23	89.575\\
24	88.495\\
25	86.86\\
26	90.675\\
27	91.155\\
28	90.945\\
29	90.155\\
30	89.135\\
};
\addlegendentry{0.5: (10,5)}

\addplot [color=mycolor2, dotted, line width=1.0pt]
  table[row sep=crcr]{%
1	36.575\\
2	53.605\\
3	61.17\\
4	64.44\\
5	66.53\\
6	80.195\\
7	81.8\\
8	82.035\\
9	81.045\\
10	79.605\\
11	86.16\\
12	86.885\\
13	86.575\\
14	85.405\\
15	83.265\\
16	89.58\\
17	89.595\\
18	88.7\\
19	87.625\\
20	85.41\\
21	89.83\\
22	90.39\\
23	90.17\\
24	89.77\\
25	88.84\\
26	91.335\\
27	91.175\\
28	90.62\\
29	90.09\\
30	88.84\\
};
\addlegendentry{0.5: (15,5)}

\addplot [color=mycolor3, line width=1.0pt, mark=o, mark options={solid, mycolor3}]
  table[row sep=crcr]{%
1	35.83\\
2	53.85\\
3	61.29\\
4	66\\
5	77.18\\
6	80.84\\
7	81.48\\
8	82.62\\
9	84.835\\
10	85.965\\
11	87.165\\
12	87.375\\
13	88.13\\
14	88.28\\
15	88.455\\
16	89.29\\
17	89.49\\
18	89.73\\
19	90.395\\
20	90.69\\
21	90.54\\
22	90.62\\
23	90.77\\
24	90.865\\
25	91.12\\
26	91.065\\
27	91.03\\
28	91.43\\
29	91.31\\
30	91.4\\
};
\addlegendentry{0.5: $\bar{A}=2$}

\addplot [color=mycolor4, line width=1.0pt]
  table[row sep=crcr]{%
1	23.25\\
2	32.96\\
3	39.485\\
4	43.005\\
5	44.12\\
6	59.175\\
7	60.45\\
8	59.88\\
9	56.42\\
10	54.445\\
11	69.965\\
12	71.11\\
13	70.08\\
14	65.6\\
15	62.12\\
16	75.315\\
17	76.745\\
18	76.585\\
19	73.33\\
20	71.28\\
21	82.585\\
22	82.34\\
23	81.365\\
24	80.08\\
25	76.2\\
26	84.51\\
27	85.055\\
28	83.475\\
29	79.815\\
30	76.51\\
};
\addlegendentry{0.1: (10,5)}

\addplot [color=mycolor5, dashdotted, line width=1.0pt]
  table[row sep=crcr]{%
1	23.85\\
2	33.275\\
3	39.245\\
4	43.13\\
5	44.81\\
6	58.82\\
7	60.52\\
8	58.805\\
9	55.795\\
10	54.81\\
11	66.5\\
12	67.9\\
13	65.835\\
14	64.845\\
15	63.515\\
16	77.425\\
17	77.165\\
18	77.355\\
19	74.64\\
20	72.21\\
21	80.25\\
22	81.07\\
23	79.86\\
24	75.665\\
25	73.56\\
26	80.79\\
27	81.435\\
28	80.67\\
29	78.965\\
30	74.48\\
};
\addlegendentry{0.1: (15,5)}

\addplot [color=mycolor6, dashdotted, line width=1.0pt, mark=o, mark options={solid, mycolor6}]
  table[row sep=crcr]{%
1	22.785\\
2	33.63\\
3	40.385\\
4	41.28\\
5	54.625\\
6	59.475\\
7	59.33\\
8	61.335\\
9	64.135\\
10	67.13\\
11	69.025\\
12	71.405\\
13	72.185\\
14	71.11\\
15	72.21\\
16	74.855\\
17	77.405\\
18	78.77\\
19	79.285\\
20	80.93\\
21	81.535\\
22	81.225\\
23	81.955\\
24	81.57\\
25	81.855\\
26	83.345\\
27	83.14\\
28	83.34\\
29	84.685\\
30	84.875\\
};
\addlegendentry{0.1: $\bar{A}=2$}

\end{axis}
\end{tikzpicture}%
		\caption{$\{P_{k,i}^{(t)}\}_{i=0}^2=[0.0005, 0.1,0,3]$.}
		\label{fig:dynFixComp1}
	\end{subfigure}
	\caption{Test accuracy comparison of A-MTDC ($\bar{A}=2$) and vanilla MTDC, with patterns $(10,5)$ and $(15,5)$, for device scheduling ratios $0.5$ and $0.1$.}
	%\label{fig:histogram}
\end{figure}

\subsection{Advantages of Age-Aware Device Scheduling}

We demonstrate the performance gain of our proposed age-aware scheduling over the baseline random scheduling, and over a state-of-the-art version-age-based scheduling policy \cite{hu2024ver}\footnote{This method prioritizes devices with low participation frequency, and therefore minimizes the overall device staleness in the system.}, in the presence of decoding failures. 
We consider the same setting in Section \ref{sec:firstSim}. As shown in Figure \ref{fig:diffSch}, our age-aware scheduling method generally outperforms the others for FL frameworks with differential model broadcasts. Our method prioritizes devices with fresher model updates and achieves better test accuracy than the version-age-based policy. This suggests that, for scheduling design, excluding stale updates is more important than ensuring fairness of device participation when the system is subject to decoding failures. 
On the other hand, when the overall situation of model outdatedness is mild, e.g., when \emph{AllFull} is adopted, all three methods will have similar device scheduling probabilities, which is reflected by their similar learning performance. 
\begin{figure}[t]
	%\centering
	\vspace{-4.5cm}
	\input{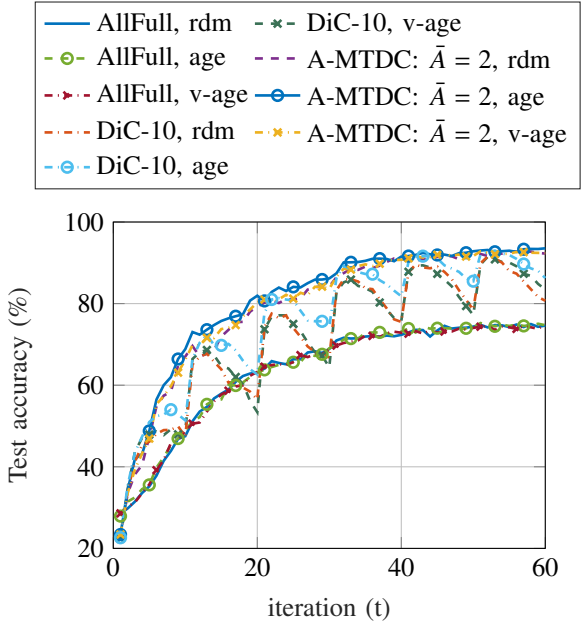}
	\caption{Learning performance with random scheduling ('rdm'), version-age-based scheduling ('v-age'), and the proposed age-aware scheduling scheme ('age') for different     schemes.}
	\label{fig:diffSch}
\end{figure}

\section{Conclusions and Future Work}
This paper proposed a mixed-timescale differential coding (MTDC) framework for DL transmission of global models in FL systems, leveraging the temporal correlation among model iterates. Compared to always broadcasting the full model and to conventional DiC schemes, MTDC achieves competitive learning performance while significantly improving communication efficiency and robustness to DL decoding failures.

We established convergence guarantees for FedAvg under the proposed MTDC framework, revealing how model staleness induced by decoding failures can slow the per-iteration contraction toward the optimum. These insights motivated the design of an age-aware extension, A-MTDC, along with an age-aware device scheduling policy. The effectiveness of the proposed schemes was validated through simulations, which demonstrate consistent learning performance gains over existing methods under comparable communication budgets.

The MTDC framework inherently involves trade-offs between the frequency of full-model broadcasts and that of first-level differential updates. Optimally balancing these transmissions is highly scenario-dependent (e.g., on the DL decoding failure statistics) and a detailed quantitative optimization is therefore left for future work. On the other hand, while this work adopts scalar quantization, further performance improvement may be achievable by incorporating more advanced quantization techniques, such as vector quantization \cite{oh2024fedvqcs, vector-quant-FL}. 
Finally, the MTDC mechanism is also applicable to decentralized FL frameworks, which requires more complex modeling and design considerations. We leave this part to our future studies.

%\begin{figure*}[h]
%	\centering
%	\includegraphics[scale=0.5]{flowchart}
%	\caption{The flowchart from the perspective of device $k$ at iteration $t$, starting from the model being broadcast by the server and ending with either model update transmission, reception failure, or iteration skipping.}
%	\label{fig:flowChart}
%\end{figure*}

%\begin{figure*}[h]
%	\centering
%	\includegraphics[scale=0.5]{broadcastInfo}
%	\caption{The decision flow at the server for computing $\hat{\boldsymbol{\theta}}(t)$, followed by memory buffer updates at $\boldsymbol{\theta}^{(i)},\forall i$.}
%	\label{fig:broadcastInfo}
%\end{figure*}
\appendix
We introduce the following auxiliary variables:
\begin{itemize}
	\item the aggregated gradient updates
	\begin{equation*}
		\boldsymbol{G}(\boldsymbol{t})=\sum\nolimits_{k\in\mathcal{K}}w_k\nabla F_k(\tilde{\boldsymbol{\theta}}(t_k)),
	\end{equation*}
	where $\boldsymbol{t}=[t_1,...,t_K]^{T}$ records the timestamps of the global model at each device, and
	\item the distance of the $t$-th iterate $\tilde{\boldsymbol{\theta}}(t)$ to the optimum
	\begin{equation*}
		\delta_t=\|\tilde{\boldsymbol{\theta}}(t)-\boldsymbol{\theta}^*\|.
	\end{equation*}
\end{itemize}
 Useful lemmas and inequalities  can be found in Appendix \ref{sec:more_lemma}.

\subsection{Proof of Theorem \ref{thm:main}}
\label{apx:main_thm}

We investigate how $\|\tilde{\boldsymbol{\theta}}(t+1)-\boldsymbol{\theta}^*\|^2=\delta_{t+1}^2$ evolves over time
in different scenarios, i.e., when $t+1\in\mathcal{T}_0$, $\mathcal{T}_1$, and $\mathcal{T}_2$.
\subsubsection{If $t+1\in\mathcal{T}_2$, a Second-Level Differential Update}
Based on \eqref{eq:mdl_rcst}, \eqref{eq:t_ref}, and \eqref{eq:h_hat_def}, $$\tilde{\boldsymbol{\theta}}(t+1)=\tilde{\boldsymbol{\theta}}(t)+\hat{\boldsymbol{\theta}}(t+1)=\tilde{\boldsymbol{\theta}}(t)+Q_2[\boldsymbol{\theta}(t+1)-\tilde{\boldsymbol{\theta}}(t)].$$ Furthermore with \eqref{eq:update_MTDC}, \eqref{eq:grdtAge}, and $$\{t-a_k^{(t)}\}_{\forall k}=[t-a_1^{(t)},...,t-a_K^{(t)}]^T,$$ we have (hereafter, $\boldsymbol{1}=[1,...,1]^T$)
\begin{subequations}
	\label{ieq:mainIeq}
	\begin{align}
		\delta_{t+1}^2&=\Big|\Big|\tilde{\boldsymbol{\theta}}(t)-Q_2\left[\eta\boldsymbol{G}\left(\{t-a_k^{(t)}\}_{\forall k}\right)\right]-\boldsymbol{\theta}^*\Big|\Big|^2\nonumber\\
		&=\delta_t^2-2\eta\left(\tilde{\boldsymbol{\theta}}(t)-\boldsymbol{\theta}^*\right)^T\boldsymbol{G}(t\boldsymbol{1})\label{eq:eq0}\\
		&\quad+\Big|\Big|Q_2\left[\eta\boldsymbol{G}\left(\{t-a_k^{(t)}\}_{\forall k}\right)\right]\Big|\Big|^2\label{eq:eq1}\\
		&\quad+2\left[\tilde{\boldsymbol{\theta}}(t)-\boldsymbol{\theta}^*\right]^T\Bigg\{\eta\boldsymbol{G}(t\boldsymbol{1})- Q_2\left[\eta\boldsymbol{G}\left(\{t-a_k^{(t)}\}_{\forall k}\right)\right]\Bigg\}.\label{eq:eq3}
	\end{align}
\end{subequations}
First, based on Assumption \ref{asp:conv} and using that $F(\tilde{\boldsymbol{\theta}}(t))\geq F(\boldsymbol{\theta}^*)$,
\begin{align}
	\text{\eqref{eq:eq0}}&\leq\delta_t^2-2\eta\sum\nolimits_{k\in\mathcal{K}}w_k\left[F_k\left(\tilde{\boldsymbol{\theta}}(t)\right)-F_k\left(\boldsymbol{\theta}^*\right)+\mu\delta_t^2/2\right]\nonumber\\
	&\leq\left(1-\mu\eta\right)\delta_t^2,\label{ieq:firstTermMu}
\end{align}
and therefore after taking total expectation,
\begin{align}
	&\mathbb{E}\left[\text{\eqref{eq:eq0}}\right]\leq
	\left(1-\mu\eta\right)\mathbb{E}\left[\delta_t^2\right].\label{ieq:eq0_e}
\end{align}
For \eqref{eq:eq1}, we first evaluate the conditional expectation based on a realization up to iteration $t$, and then take expectation:
\begin{align}
	&\mathbb{E}\left[\text{\eqref{eq:eq1}}\right]=\mathbb{E}\left[\mathbb{E}\left[\text{\eqref{eq:eq1}}\Big|\left\{\tilde{\boldsymbol{\theta}}(t-a_k^{(t)})\right\}_{k=1}^K\right]\right]\nonumber\\
	&\leq (\sigma_2+1)\eta^2\mathbb{E}\left[\Big|\Big|\boldsymbol{G}\left(\{t-a_k^{(t)}\}_{\forall k}\right)\Big|\Big|^2\right]\label{ieq:applyQassump}\\
	&\leq2L^2\eta^2(\sigma_2+1)\left[\sum_{k\in\mathcal{K}}w_k\mathbb{E}\left[\delta_{t-a_k^{(t)}}^2\right]+\sum_{k\in\mathcal{K}}w_k\|\boldsymbol{\theta}^*_k-\boldsymbol{\theta}^*\|^2\right]\label{ieq_L_sepNorm}\\
	&\leq L^2\eta^2\hat{\sigma}\Bigg\{\sum\nolimits_{k\in\mathcal{K}}w_k\mathbb{E}\left[\delta_{t-a_k^{(t)}}^2\right]+\zeta\Bigg\},\label{ieq:b_p1}
\end{align}
where \eqref{ieq:applyQassump} follows from \eqref{eq:rdmqnt} as
\begin{align}
	\mathbb{E}\left[\|Q_2(\boldsymbol{\theta})\|^2|\boldsymbol{\theta}\right]&=\mathbb{E}\left[\|Q_2(\boldsymbol{\theta})-\boldsymbol{\theta}\|^2|\boldsymbol{\theta}\right]+\|\boldsymbol{\theta}\|^2\nonumber\\
	&=\left(\sigma_2+1\right)\|\boldsymbol{\theta}\|^2,\label{eq:qntIntept}
\end{align}
and $Q_2(\eta\boldsymbol{x})$ has the same probability distribution as $\eta Q_2(\boldsymbol{x})$;
\eqref{ieq_L_sepNorm} follows from \eqref{ieq:wk_moveout}, Assumption \ref{asp:smth}, and \eqref{ieq:norm_ieq}; \eqref{ieq:b_p1} is based on \eqref{eq:zeta_hetero} and \eqref{eq:sigma}.\footnote{Introducing the  constant $\hat{\sigma}$ yields a simpler bound, though it may not be the tightest.} 
We handle \eqref{eq:eq3} with a similar approach (evaluating conditional expectation and using $\mathbb{E}\left[Q_2(\boldsymbol{\theta})\right]=\boldsymbol{\theta}$):
\begin{align}
	&\mathbb{E}\left[\text{\eqref{eq:eq3}}\right]=\mathbb{E}\left[\mathbb{E}\left[\text{\eqref{eq:eq3}}\big|\tilde{\boldsymbol{\theta}}(t),\left\{\tilde{\boldsymbol{\theta}}(t-a_k^{(t)})\right\}_{\forall k}\right]\right]\nonumber\\
	&=2\eta\mathbb{E}\Bigg\{\left(\tilde{\boldsymbol{\theta}}(t)-\boldsymbol{\theta}^*\right)^T\left[\boldsymbol{G}\left(t\boldsymbol{1}\right)-\boldsymbol{G}\left(\{t-a_k^{(t)}\}_{\forall k}\right)\right]\Bigg\}\nonumber\\
	&\leq2\eta\mathbb{E}\Bigg\{\delta_t\Big|\Big|\boldsymbol{G}\left(t\boldsymbol{1}\right)-\boldsymbol{G}\left(\{t-a_k^{(t)}\}_{\forall k}\right)\Big|\Big|\Bigg\},\label{ieq:crossAge}\\
    &\leq2L^2\hat{\sigma}a_{\lim}\eta^2\Bigg[2a_{\lim}\sqrt{\zeta}+a_{\lim}\left(2\sqrt{\zeta}+3/2\right)\mathbb{E}\left[\delta_t^2\right]\nonumber\\
	&\quad\quad\quad\quad\quad\quad+\sum\nolimits_{i=1}^{\min\left(t-1,3a_{\lim}\right)}\mathbb{E}\left[\delta_{t-i}^2\right]/2
    \Bigg]\label{ieq:lm2Aply}
\end{align}
where \eqref{ieq:crossAge} follows from Cauchy–Schwarz inequality and \eqref{ieq:lm2Aply} is based on Lemma \ref{lm:outdateCrossAge} (given in Appendix \ref{sec:more_lemma}).

Combining \eqref{ieq:eq0_e}, \eqref{ieq:b_p1}, and \eqref{ieq:lm2Aply},
\begin{align}
	\mathbb{E}\left[\delta_{t+1}^2\right]&\leq \left[1-\mu\eta+2L^2\eta^2\hat{\sigma}a_{\lim}^2\left(2\sqrt{\zeta}+3/2\right)\right]\mathbb{E}\left[\delta_t^2\right]\nonumber\\
	&\quad+L^2\eta^2\hat{\sigma}\Bigg\{\sqrt{\zeta}\left(\sqrt{\zeta}+4a_{\lim}^2\right)+\sum\nolimits_{k\in\mathcal{K}}w_k\mathbb{E}\left[\delta_{t-a_k^{(t)}}^2\right]\nonumber\\
	&\quad\quad\quad\quad\quad+a_{\lim}\sum\nolimits_{i=1}^{\min\left(t-1,3a_{\lim}\right)}\mathbb{E}\left[\delta_{t-i}^2\right]\Bigg\}\label{ieq:sumup}
\end{align}
\subsubsection{If $t+1\in\mathcal{T}_1$, a First-Level Differential Update}
Note that the server broadcasts second-level differential updates between iteration $r_{t+1}$ and $t+1$, i.e., $\{r_{t+1}+1,...,t\}\subset\mathcal{T}_2$. Hence, the reference model for reconstructing $\tilde{\boldsymbol{\theta}}(t+1)$ satisfies
\begin{equation}
\tilde{\boldsymbol{\theta}}\left(r_{t+1}\right)+\sum_{\tau=r_{t+1}}^{t-1}Q_2\left[-\eta\boldsymbol{G}\left(\{\tau-a_k^{(\tau)}\}_{\forall k}\right)\right]=\tilde{\boldsymbol{\theta}}(t).\label{eq:t1_ref}
\end{equation}
Define
\begin{equation}
	\boldsymbol{y}_{t+1}=\sum_{\tau=r_{t+1}}^{t-1}Q_2\left[-\eta\boldsymbol{G}\left(\{\tau-a_k^{(\tau)}\}_{\forall k}\right)\right]-\eta\boldsymbol{G}\left(\{t-a_k^{(t)}\}_{\forall k}\right).\label{eq:y_tp1_def}
\end{equation}
The broadcast model $\hat{\boldsymbol{\theta}}(t+1)=Q_1\left(\boldsymbol{\theta}(t+1)-\tilde{\boldsymbol{\theta}}(r_{t+1})\right)$ can then be expressed as
\begin{align}
	\hat{\boldsymbol{\theta}}(t+1)&=Q_1\left[\tilde{\boldsymbol{\theta}}(t)-\eta\boldsymbol{G}\left(\{t-a_k^{(t)}\}_{\forall k}\right)-\tilde{\boldsymbol{\theta}}(r_{t+1})\right]=Q_1(\boldsymbol{y}_{t+1}),\label{eq:thetaHat_1}
\end{align} 
by using $\boldsymbol{\theta}(t+1)=\tilde{\boldsymbol{\theta}}(t)-\eta\boldsymbol{G}\left(\{t-a_k^{(t)}\}_{\forall k}\right)$, \eqref{eq:t1_ref}, and \eqref{eq:y_tp1_def}.  
Since $\tilde{\boldsymbol{\theta}}(t+1)=\tilde{\boldsymbol{\theta}}(r_{t+1})+\hat{\boldsymbol{\theta}}(t+1)$, by using \eqref{eq:t1_ref} and \eqref{eq:thetaHat_1},
\begin{subequations}
	\label{eq:caseT1}
	\begin{align}
		&\delta_{t+1}^2\nonumber\\
		&=\Big|\Big|\tilde{\boldsymbol{\theta}}(t)-\sum_{\tau=r_{t+1}}^{t-1}Q_2\left[-\eta\boldsymbol{G}\left(\{\tau-a_k^{(\tau)}\}_{\forall k}\right)\right]+Q_1\left(\boldsymbol{y}_{t+1}\right)-\boldsymbol{\theta}^*\Big|\Big|^2\nonumber\\
		&=\Big|\Big|\tilde{\boldsymbol{\theta}}(t)-\boldsymbol{\theta}^*+Q_1\left(\boldsymbol{y}_{t+1}\right)-\boldsymbol{y}_{t+1}-\eta\boldsymbol{G}\left(\{t-a_k^{(t)}\}_{\forall k}\right)\Big|\Big|^2\nonumber\\
		&=\delta_t^2-2\eta\left(\tilde{\boldsymbol{\theta}}(t)-\boldsymbol{\theta}^*\right)^T\boldsymbol{G}(t\boldsymbol{1})+\eta^2\Big|\Big|\boldsymbol{G}\left(\{t-a_k^{(t)}\}_{\forall k}\right)\Big|\Big|^2\label{termLegacy}\\
		&\quad+\|Q_1\left(\boldsymbol{y}_{t+1}\right)-\boldsymbol{y}_{t+1}\|^2\label{termNew}\\
		&\quad+2\eta\left(\tilde{\boldsymbol{\theta}}(t)-\boldsymbol{\theta}^*\right)^T\left[\boldsymbol{G}\left(t\boldsymbol{1}\right)-\boldsymbol{G}\left(\{t-a_k^{(t)}\}_{\forall k}\right)\right]\label{termOld}\\
		&\quad+2\left[Q_1\left(\boldsymbol{y}_{t+1}\right)-\boldsymbol{y}_{t+1}\right]^T\Big\{\tilde{\boldsymbol{\theta}}(t)-\boldsymbol{\theta}^*-\eta\boldsymbol{G}\left(\{t-a_k^{(t)}\}_{\forall k}\right)\Big\}.\label{ieq:y2}
	\end{align}
\end{subequations}
By applying \eqref{ieq:firstTermMu}, \eqref{ieq:wk_moveout}, Assumption \ref{asp:smth}, \eqref{ieq:norm_ieq}, and \eqref{eq:zeta_hetero},
\begin{align}
	\text{\eqref{termLegacy}}&\leq\left(1-\mu\eta\right)\delta_t^2+2L^2\eta^2\sum\nolimits_{k\in\mathcal{K}}w_k\delta_{t-a_k^{(t)}}^2+2L^2\eta^2\zeta.\nonumber
\end{align}
By taking total expectation of \eqref{termLegacy},
\begin{align}
	\mathbb{E}\left[\text{\eqref{termLegacy}}\right]&\leq\left(1-\mu\eta\right)\mathbb{E}\left[\delta_t^2\right]+2L^2\eta^2\left\{\sum_{k\in\mathcal{K}}w_k\mathbb{E}\left[\delta_{t-a_k^{(t)}}^2\right]+\zeta\right\}.\label{ieq:1stTerm}
\end{align}
For \eqref{termNew}, we first deal with the randomness of $Q_1(\cdot)$,
\begin{align}
	&\mathbb{E}\left[\text{\eqref{termNew}}\right]\nonumber\\
	&=\mathbb{E}\Bigg[\mathbb{E}\left[\text{\eqref{termNew}}\Big|\left\{\tilde{\boldsymbol{\theta}}(\tau-a_k^{(\tau)})\big|r_{t+1}\leq\tau\leq t,\forall k\right\},Q_2\left(\cdot\right)\right]\Bigg]\nonumber\\
	&\leq\sigma_1\mathbb{E}\left[\|\boldsymbol{y}_{t+1}\|^2\right]\label{ieq:conditionalSame}\\
	&\leq\sigma_1\mathbb{E}\Bigg\{(t+1-r_{t+1})\Big[\eta^2\left\|\boldsymbol{G}\left(\{t-a_k^{(t)}\}_{\forall k}\right)\right\|^2\nonumber\\
	&\quad\quad\quad\quad+\sum\nolimits_{\tau=r_{t+1}}^{t-1}\left\|Q_2\left[-\eta\boldsymbol{G}\left(\{\tau-a_k^{(\tau)}\}_{\forall k}\right)\right]\right\|^2\Big]\Bigg\}\label{ieq:normSplit}\\
	&\leq a_{\lim}\sigma_1\eta^2L^2\Bigg\{
	\sum\nolimits_{k\in\mathcal{K}}w_k\mathbb{E}~\|\tilde{\boldsymbol{\theta}}(t-a_k^{(t)})-\boldsymbol{\theta}_k^*\|^2\nonumber\\
	&\quad+\left(\sigma_2+1\right)\mathbb{E}\left[\sum_{\tau=r_{t+1}}^{t-1}\sum_{k\in\mathcal{K}}w_k\|\tilde{\boldsymbol{\theta}}(\tau-a_k^{(\tau)})-\boldsymbol{\theta}_k^*\|^2\right]\Bigg\}\label{ieq:wkOutLQ2}\\
	&\leq2a_{\lim}\sigma_1\eta^2L^2\Bigg\{\zeta\left[\left(\sigma_2+1\right)a_{\lim}+1\right]+
	\sum\nolimits_{k\in\mathcal{K}}w_k\mathbb{E}\left[\delta_{t-a_k^{(t)}}^2\right]\nonumber\\
	&\quad\quad\quad\quad+	\left(\sigma_2+1\right)\sum\nolimits_{k\in\mathcal{K}}w_k\mathbb{E}\left[\sum\nolimits_{\tau=r_{t+1}}^{t-1}\delta_{\tau-a_k^{(\tau)}}^2\right]\Bigg\},\label{ieq:2ndTerm}
\end{align}
where \eqref{ieq:conditionalSame} evaluates the expectation over the random quantizer $Q_1(\cdot)$ according to \eqref{eq:rdmqnt}, conditioned on the previously reconstructed models and $Q_2$-quantized components in $\boldsymbol{y}_{t+1}$; \eqref{ieq:normSplit} is based on \eqref{eq:y_tp1_def} and \eqref{ieq:norm_ieq};
\eqref{ieq:wkOutLQ2} follows from applying \eqref{ieq:wk_moveout}, Assumption \ref{asp:smth}, and \eqref{eq:qntIntept}; and \eqref{ieq:2ndTerm} is obtained by applying \eqref{ieq:norm_ieq} and \eqref{eq:zeta_hetero}.
For \eqref{termOld}, we have
\begin{align}
	&\mathbb{E}\left[\text{\eqref{termOld}}\right]\leq2\eta\mathbb{E}\left[\delta_t\cdot\Big|\Big|\boldsymbol{G}(t\boldsymbol{1})-\boldsymbol{G}\left(\{t-a_k^{(t)}\}_{\forall k}\right)\Big|\Big|\right]\nonumber\\
	&\leq2L^2\hat{\sigma}a_{\lim}\eta^2\Bigg[2a_{\lim}\sqrt{\zeta}+a_{\lim}\left(2\sqrt{\zeta}+3/2\right)\mathbb{E}\left[\delta_t^2\right]\nonumber\\
	&\quad\quad\quad\quad\quad\quad+\sum\nolimits_{i=1}^{\min\left(t-1,3a_{\lim}\right)}\mathbb{E}\left[\delta_{t-i}^2\right]/2\Bigg],\label{ieq:3rdTerm}
\end{align}
based on Cauchy–Schwarz inequality and Lemma \ref{lm:outdateCrossAge}.
For \eqref{ieq:y2}, since $Q_1(\cdot)$ is unbiased, as given in \eqref{eq:rdmqnt}, we have\footnote{The expectation is conditioned on those $Q_2$-quantized components in $\boldsymbol{y}_{t+1}$.}
\begin{equation*}
	\mathbb{E}\left[\text{\eqref{ieq:y2}}\Big|\tilde{\boldsymbol{\theta}}(t),\{\tilde{\boldsymbol{\theta}}(t-a_k^{(t)})| 1\leq k\leq K\}, Q_2\left(\cdot\right)\right]=0.
\end{equation*}
Together with \eqref{ieq:1stTerm}, \eqref{ieq:2ndTerm}, \eqref{ieq:3rdTerm},
$\hat{\sigma}\geq 2\sigma_2\geq2\sigma_1$, and $\hat{\sigma}\geq\sigma_1(\sigma_2+1)$,
we have
\begin{align}
	\mathbb{E}\left[\delta_{t+1}^2\right]&\leq \left(1-\mu\eta+2L^2\hat{\sigma}a_{\lim}^2\eta^2\left(2\sqrt{\zeta}+3/2\right)\right)\mathbb{E}\left[\delta_t^2\right]\nonumber\\
	&\quad+L^2\eta^2\sqrt{\zeta}\left[\sqrt{\zeta}\left(2+\hat{\sigma}a_{\lim}\left(2a_{\lim}+1\right)\right)+4\hat{\sigma}a_{\lim}^2\right]\nonumber\\
	&\quad+L^2\eta^2\left(2+\hat{\sigma}a_{\lim}\right)\sum\nolimits_{k\in\mathcal{K}}w_k\mathbb{E}\left[\delta_{t-a_k^{(t)}}^2\right]\nonumber\\
	&\quad+2a_{\lim}\hat{\sigma}L^2\eta^2\sum\nolimits_{k\in\mathcal{K}}w_k\mathbb{E}\left[\sum\nolimits_{\tau=r_{t+1}}^{t-1}\delta_{\tau-a_k^{(\tau)}}^2\right]\nonumber\\
	&\quad+L^2\hat{\sigma}a_{\lim}\eta^2\sum\nolimits_{i=1}^{\min\left(t-1,3a_{\lim}\right)}\mathbb{E}\left[\delta_{t-i}^2\right].\label{ieq:caseFirstLvl}
\end{align}

\subsubsection{If $t+1\in\mathcal{T}_0$}
The upper bound of $\delta_{t+1}^2$ can be similarly derived as in the case of $t+1\in\mathcal{T}_2$. Specifically,
\begin{itemize}
    \item \eqref{ieq:mainIeq} holds with
     $Q_2\left[\eta\boldsymbol{G}\left(\{t-a_k^{(t)}\}_{\forall k}\right)\right]$ replaced by $\eta\boldsymbol{G}\left(\{t-a_k^{(t)}\}_{\forall k}\right)$;
     \item \eqref{ieq:b_p1} holds with $\hat{\sigma}$ replaced by a scaling $2$;
     \item \eqref{ieq:eq0_e} and \eqref{ieq:lm2Aply} hold without any change.
\end{itemize}
This gives
\begin{align}
	&\mathbb{E}\left[\delta_{t+1}^2\right]\leq\left[1-\mu\eta+2L^2\eta^2\hat{\sigma}a_{\lim}^2\left(2\sqrt{\zeta}+3/2\right)\right]\mathbb{E}\left[\delta_t^2\right]\nonumber\\
	&+L^2\eta^2\Bigg\{2\sqrt{\zeta}\left(\sqrt{\zeta}+2a_{\lim}^2\hat{\sigma}\right)+2\sum\nolimits_{k\in\mathcal{K}}w_k\mathbb{E}\left[\delta_{t-a_k^{(t)}}^2\right]\nonumber\\
	&\quad\quad\quad+a_{\lim}\hat{\sigma}\sum\nolimits_{i=1}^{\min\left(t-1,3a_{\lim}\right)}\mathbb{E}\left[\delta_{t-i}^2\right]\Bigg\}.\label{ieq:sumup0}
\end{align}
Comparing \eqref{ieq:sumup}, \eqref{ieq:caseFirstLvl}, and \eqref{ieq:sumup0}, we conclude that \eqref{ieq:caseFirstLvl} holds for all $t$, as it has the highest upper bound.

\subsubsection{Overall Convergence Bound}
Based on \eqref{ieq:stepSize_ieq},
\begin{align*}	&L^2\eta^2\left[2\hat{\sigma}a_{\lim}^2\left(2\sqrt{\zeta}+3/2\right)+2+\hat{\sigma}a_{\lim}+2\hat{\sigma}a_{\lim}^2+3\hat{\sigma}a_{\lim}^2\right]\nonumber\\
	&+1-\mu\eta<1-\eta\mu/2\triangleq C.
\end{align*}
Then, \eqref{ieq:caseFirstLvl} can be rearranged as\footnote{In \eqref{ieq:caseFirstLvl}, the $\tau$-summation has no more than $a_{\lim}$ terms and the $i$-summation has at most $3a_{\lim}$ terms, respectively.}
\begin{equation}
	\mathbb{E}\left[\delta_{t+1}^2\right]\leq C\max\left(\mathbb{E}\left[\delta_{t}^2\right],...,\mathbb{E}\left[\delta_{\max\left(1,t-3a_{\lim}\right)}^2\right]\right)+\eta^2\epsilon,\label{ieq:maxIneq}
\end{equation}
where $\epsilon$ is defined in \eqref{eq:epsilon}.
We will prove by induction that
\begin{equation}
	\mathbb{E}\left[\delta_{t+1}^2\right]\leq C^{\lfloor\frac{t-1}{3a_{\lim}+1}\rfloor+1}\mathbb{E}\left[\delta_{1}^2\right]+\eta^2\epsilon\sum\nolimits_{i=0}^{t-1}C^i. \label{ieq:claimed}
\end{equation}
When $t=1$, \eqref{ieq:maxIneq} gives $\mathbb{E}\left[\delta_{2}^2\right]\leq C\mathbb{E}\left[\delta_{1}^2\right]+\eta^2\epsilon$. Assuming that \eqref{ieq:claimed} holds, we evaluate $\mathbb{E}\left[\delta_{t+2}^2\right]$ based on \eqref{ieq:maxIneq}:
\begin{align}
	&\mathbb{E}\left[\delta_{t+2}^2\right]\leq C\max\left(\mathbb{E}\left[\delta_{t+1}^2\right],...,\mathbb{E}\left[\delta_{\max\left(1,t+1-
    3a_{\lim}\right)}^2\right]\right)+\eta^2\epsilon\nonumber\\
	&\leq C\left[C^{\lfloor\frac{\max\left(1,t+1-3a_{\lim}\right)-2}{3a_{\lim}+1}\rfloor+1}\mathbb{E}\left[\delta_{1}^2\right]+\eta^2\epsilon\sum\nolimits_{i=0}^{t-1}C^i\right]+\eta^2\epsilon\label{ieq:sep_step}\\
	&=C^{\lfloor\frac{\max\left(3a_{\lim},t\right)}{3a_{\lim}+1}\rfloor+1}\mathbb{E}\left[\delta_{1}^2\right]+\eta^2\epsilon\sum\nolimits_{i=0}^{t}C^i\label{sep_step_fin}
\end{align}
where \eqref{ieq:sep_step} holds because  the first and  second terms in \eqref{ieq:claimed} decreases and increases with $t$, respectively.
Note that
\begin{equation*}
	\left\lfloor\frac{\max\left(3a_{\lim},t\right)}{3a_{\lim}+1}\right\rfloor=\begin{cases}
		\lfloor\frac{3a_{\lim}}{3a_{\lim}+1}\rfloor=\lfloor\frac{t}{3a_{\lim}+1}\rfloor,&t<3a_{\lim}\\
		\lfloor\frac{t}{3a_{\lim}+1}\rfloor,&t\geq 3a_{\lim}
	\end{cases}.
\end{equation*} 
Then
\eqref{sep_step_fin} can be rewritten as
\begin{equation*}
	\mathbb{E}\left[\delta_{t+2}^2\right]\leq C^{\lfloor\frac{t}{3a_{\lim}+1}\rfloor+1}\mathbb{E}\left[\delta_{1}^2\right]+\eta^2\epsilon\sum\nolimits_{i=0}^{t}C^i,
\end{equation*}
which completes the proof. Finally, \eqref{ieq:optGapMain} follows from \eqref{ieq:claimed} and $\sum_{i=0}^{t-1}C^i<1/\left(1-C\right)$.

\subsection{Useful Inequalities and Lemmas}
\label{sec:more_lemma}
Let $\boldsymbol{a}_k\in\mathbb{R}^d, \forall k\in\mathcal{K}$. Based on Jensen's inequality,
	\begin{equation}
		\left\|\sum\nolimits_{k\in\mathcal{K}}w_k\boldsymbol{a}_k\right\|^p\leq\sum\nolimits_{k\in\mathcal{K}}w_k \left\|\boldsymbol{a}_k\right\|^p, p=1,2.\label{ieq:wk_moveout}
	\end{equation}
    As a special case,
	\begin{equation}
		\left\|\sum\nolimits_{i=1}^N\boldsymbol{a}_i\right\|^2\leq N\sum\nolimits_{i=1}^N\|\boldsymbol{a}_i\|^2\label{ieq:norm_ieq}.
	\end{equation}

\begin{lemma}
	The following result holds,
	\begin{align*}
		&\mathbb{E}\Big[\delta_t\Big|\Big|\boldsymbol{G}(t\boldsymbol{1})-\boldsymbol{G}\left(\{t-a_k^{(t)}\}_{\forall k}\right)\Big|\Big|\Big]\leq L^2\hat{\sigma}a_{\lim}\eta\Big[2a_{\lim}\sqrt{\zeta}\\
		&\quad\quad+a_{\lim}\left(2\sqrt{\zeta}+3/2\right)\mathbb{E}\left[\delta_t^2\right]+\sum\nolimits_{i=1}^{\min\left(t-1,3a_{\lim}\right)}\mathbb{E}\left[\delta_{t-i}^2\right]/2\Big].
	\end{align*}
	\label{lm:outdateCrossAge}
\end{lemma}
\noindent\emph{Proof.}
By applying \eqref{ieq:wk_moveout}, Assumption \ref{asp:smth}, and the triangle inequality,
\begin{align}
	&\mathbb{E}\Bigg[\delta_t\Big|\Big|\boldsymbol{G}(t\boldsymbol{1})-\boldsymbol{G}\left(\{t-a_k^{(t)}\}_{\forall k}\right)\Big|\Big|\Bigg]\nonumber\\
    &\leq L\mathbb{E}\left\{\delta_t\sum\nolimits_{k\in\mathcal{K}}w_k\|\tilde{\boldsymbol{\theta}}(t)-\tilde{\boldsymbol{\theta}}(t-a_k^{(t)})\|\right\}\nonumber\\
	&\leq L\mathbb{E}\left\{\delta_t\sum_{k\in\mathcal{K}}w_k\left[\sum\nolimits_{i=1}^{a_k^{(t)}}\|\tilde{\boldsymbol{\theta}}(t-i+1)-\tilde{\boldsymbol{\theta}}(t-i)\|\right]\right\}.\label{ieq:c_part1}
  \end{align}
%where we evaluate $\mathbb{E}[\cdot]$ in two steps: 1) up to model broadcast at iteration $t$, 2) memory update at iteration $t$, i.e., $\mathbb{E}[\cdot]=\mathbb{E}_{t,\text{mem}}[\mathbb{E}_{\leq t,\text{dl}}[\cdot]]$. 
We first handle the partial terms
\begin{equation}
	\delta_t\|\tilde{\boldsymbol{\theta}}(t-i+1)-\tilde{\boldsymbol{\theta}}(t-i)\|, \quad i\in\{1,...,a_k^{(t)}\}\label{eq:partialTerms}
\end{equation} in \eqref{ieq:c_part1} as follows.
We define $\tau_i=t-i,\forall i$. Then, \eqref{eq:partialTerms} becomes
\begin{equation*}
	\delta_t\|\tilde{\boldsymbol{\theta}}(\tau_i+1)-\tilde{\boldsymbol{\theta}}(\tau_i)\|, \quad i\in\{1,...,a_k^{(t)}\}.
\end{equation*}

\subsubsection{If a First-Level Differential Update Is Broadcast at Iteration $\tau_i+1$}
$\hat{\boldsymbol{\theta}}(\tau_i+1)=Q_1\left(\boldsymbol{\theta}(\tau_i+1)-\tilde{\boldsymbol{\theta}}(r_{\tau_i+1})\right)=Q_1(\boldsymbol{z}_i)$, where
\begin{equation*}
	\boldsymbol{z}_i=\tilde{\boldsymbol{\theta}}(\tau_i)-\eta\boldsymbol{G}\left(\{\tau_i-a_k^{(\tau_i)}\}_{\forall k}\right)-\tilde{\boldsymbol{\theta}}(r_{\tau_i+1}).
\end{equation*}
Then, $\tilde{\boldsymbol{\theta}}(\tau_i+1)-\tilde{\boldsymbol{\theta}}(\tau_i)=\hat{\boldsymbol{\theta}}(\tau_i+1)+\tilde{\boldsymbol{\theta}}(r_{\tau_i+1})-\tilde{\boldsymbol{\theta}}(\tau_i)$ gives
\begin{align}
	&\left\|\tilde{\boldsymbol{\theta}}(\tau_i+1)-\tilde{\boldsymbol{\theta}}(\tau_i)\right\|=\left\|Q_1(\boldsymbol{z}_i)-\boldsymbol{z}_i-\eta\boldsymbol{G}\left(\{\tau_i-a_k^{(\tau_i)}\}_{\forall k}\right)\right\|\nonumber\\
	&\leq\|Q_1(\boldsymbol{z}_i)-\boldsymbol{z}_i\|+\|\eta\boldsymbol{G}\left(\{\tau_i-a_k^{(\tau_i)}\}_{\forall k}\right)\|,\label{ieq:qnt_ieq}
\end{align}
by applying the  triangle inequality.
Then, from \eqref{ieq:qnt_ieq} and evaluating the effect of $Q_1(\cdot)$,
\begin{align}
	&\mathbb{E}\Big[\delta_t\|\tilde{\boldsymbol{\theta}}(\tau_i+1)-\tilde{\boldsymbol{\theta}}(\tau_i)\|\Big]\nonumber\\
	&\leq\mathbb{E}\Bigg\{\mathbb{E}\Big[\delta_t\|Q_1(\boldsymbol{z}_i)-\boldsymbol{z}_i\|\Big|\tilde{\boldsymbol{\theta}}(\tau_i),\tilde{\boldsymbol{\theta}}(r_{\tau_i+1}),\left\{\tilde{\boldsymbol{\theta}}(\tau_i-a_k^{(\tau_i)})\right\}_{k=1}^K\Big]\Bigg\}\nonumber\\
	&\quad+\mathbb{E}\left[\delta_t\|\eta\boldsymbol{G}\left(\{\tau_i-a_k^{(\tau_i)}\}_{\forall k}\right)\|\right]\nonumber\\
	&\leq\sqrt{\sigma_1}\mathbb{E}\left[\delta_t\|\boldsymbol{z}_i\|\right]+\mathbb{E}\left[\delta_t\|\eta\boldsymbol{G}\left(\{\tau_i-a_k^{(\tau_i)}\}_{\forall k}\right)\|\right]\label{ieq:sqrtVar}\\
	&\leq\sqrt{\sigma_1}\mathbb{E}\left[\delta_t\|\tilde{\boldsymbol{\theta}}(\tau_i)-\tilde{\boldsymbol{\theta}}(r_{\tau_i+1})\|\right]\nonumber\\
	&\quad+\left(1+\sqrt{\sigma_1}\right)\eta\mathbb{E}\left[\delta_t\|\boldsymbol{G}\left(\{\tau_i-a_k^{(\tau_i)}\}_{\forall k}\right)\|\right]\label{notatRplc},
\end{align}
where \eqref{ieq:sqrtVar} follows from  applying Jensen's inequality to \eqref{eq:rdmqnt} and \eqref{notatRplc} follows from using triangle inequality on $\|\boldsymbol{z}_i\|$.
Note that $r_{\tau_i+1}+1,...,\tau_i\in\mathcal{T}_2,$ by applying triangle inequality,
\begin{align}
	&\mathbb{E}\left[\delta_t\|\tilde{\boldsymbol{\theta}}(\tau_i)-\tilde{\boldsymbol{\theta}}(r_{\tau_i+1})\|\right]\leq\sum\nolimits_{j=r_{\tau_i+1}}^{\tau_i-1}\mathbb{E}\left[\delta_t\|\tilde{\boldsymbol{\theta}}(j+1)-\tilde{\boldsymbol{\theta}}(j)\|\right]\nonumber\\
	&=\sum\nolimits_{j=r_{\tau_i+1}}^{\tau_i-1}\mathbb{E}\left\{\delta_t\left\|Q_2\left[-\eta\boldsymbol{G}\left(\{j-a_k^{(j)}\}_{\forall k}\right)\right]\right\|\right\}\nonumber\\
	&=\sum\nolimits_{j=r_{\tau_i+1}}^{\tau_i-1}\mathbb{E}\Bigg\{\nonumber\\
    &\quad\quad\quad\mathbb{E}\Big\{\delta_t\left\|Q_2\left[-\eta\boldsymbol{G}\left(\{j-a_k^{(j)}\}_{\forall k}\right)\right]\right\|\Big|\left\{\tilde{\boldsymbol{\theta}}(j-a_k^{(j)})\right\}_{k=1}^K\Big\}\Bigg\}\nonumber\\
	&\leq\eta
	\sqrt{\sigma_2+1}\sum\nolimits_{j=r_{\tau_i+1}}^{\tau_i-1}\mathbb{E}\left[\delta_t\|\boldsymbol{G}\left(\{j-a_k^{(j)}\}_{\forall k}\right)\|\right].\label{ieq:qntEqalT1}
\end{align}
Inserting \eqref{ieq:qntEqalT1} into \eqref{notatRplc},
\begin{align}
	&\mathbb{E}\Big[\delta_t\|\tilde{\boldsymbol{\theta}}(\tau_i+1)-\tilde{\boldsymbol{\theta}}(\tau_i)\|\Big]\nonumber\\
	&\leq\eta\sqrt{\sigma_1\left(\sigma_2+1\right)}\sum\nolimits_{j=r_{\tau_i+1}}^{\tau_i-1}\mathbb{E}\left[\delta_t\|\boldsymbol{G}\left(\{j-a_k^{(j)}\}_{\forall k}\right)\|\right]\nonumber\\
	&\quad+\eta\left(1+\sqrt{\sigma_1}\right)\mathbb{E}\left[\delta_t\|\boldsymbol{G}\left(\{\tau_i-a_k^{(\tau_i)}\}_{\forall k}\right)\|\right]\label{ieq:qntEqalT1_1}.
\end{align}

\subsubsection{If a Full Model Is Broadcast at Iteration $\tau_i+1$}
Based on Assumption \ref{asp:ageLim},
\begin{align}
	&\|\tilde{\boldsymbol{\theta}}(\tau_i+1)-\tilde{\boldsymbol{\theta}}(\tau_i)\|\nonumber\\
	&=\Bigg|\Bigg|Q_0\Big\{\tilde{\boldsymbol{\theta}}(\tau_i)-\eta\boldsymbol{G}\left(\{\tau_i-a_k^{(\tau_i)}\}_{\forall k}\right)\Big\}-\tilde{\boldsymbol{\theta}}(\tau_i)\Bigg|\Bigg|\nonumber\\
	&=\eta\|\boldsymbol{G}\left(\{\tau_i-a_k^{(\tau_i)}\}_{\forall k}\right)\|.\nonumber\\
	&\Rightarrow\mathbb{E}\left[\delta_t\|\tilde{\boldsymbol{\theta}}(\tau_i+1)-\tilde{\boldsymbol{\theta}}(\tau_i)\|\right]=\eta\mathbb{E}\left[\delta_t\|\boldsymbol{G}\left(\{\tau_i-a_k^{(\tau_i)}\}_{\forall k}\right)\|\right].\label{ieq:qntEqalT0}
\end{align}

\subsubsection{If a Second-Level Differential Update Is Broadcast at Iteration $\tau_i+1$}
\begin{align}
	&\mathbb{E}\left[\delta_t\|\tilde{\boldsymbol{\theta}}(\tau_i+1)-\tilde{\boldsymbol{\theta}}(\tau_i)\|\right]=\mathbb{E}\left[\delta_t\Big|\Big|Q_2\left[-\eta\boldsymbol{G}\left(\{\tau_i-a_k^{(\tau_i)}\}_{\forall k}\right)\right]\Big|\Big|\right]\nonumber\\
	&\leq\eta\sqrt{\sigma_2+1}\mathbb{E}\left[\delta_t\left\|\boldsymbol{G}\left(\{\tau_i-a_k^{(\tau_i)}\}_{\forall k}\right)\right\|\right].\label{ieq:qntEqalT2}
\end{align}
Based on \eqref{ieq:qntEqalT1_1}, \eqref{ieq:qntEqalT0}, and \eqref{ieq:qntEqalT2}, we conclude that
\begin{align}
	&\mathbb{E}\left[\delta_t\|\tilde{\boldsymbol{\theta}}(\tau_i+1)-\tilde{\boldsymbol{\theta}}(\tau_i)\|\right]\nonumber\\
	&\leq\eta(1+\sqrt{\sigma_2})\mathbb{E}\left[\delta_t\|\boldsymbol{G}\left(\{\tau_i-a_k^{(\tau_i)}\}_{\forall k}\right)\|\right]\nonumber\\
	&\quad+\eta\sqrt{\sigma_1(\sigma_2+1)}\sum\nolimits_{j=r_{\tau_i+1}}^{\tau_i-1}\mathbb{E}\left[\delta_t\|\boldsymbol{G}\left(\{j-a_k^{(j)}\}_{\forall k}\right)\|\right]\nonumber\\
	&\leq\eta\hat{\sigma} L\sum\nolimits_{j=r_{\tau_i+1}}^{\tau_i}\sum\nolimits_{k\in\mathcal{K}}w_k\mathbb{E}\left[\delta_t\|\tilde{\boldsymbol{\theta}}(j-a_k^{(j)})-\boldsymbol{\theta}_k^*\|\right],\label{ieq:qntUlDl}
\end{align}
$\tau_i+1\in\{\mathcal{T}_0,\mathcal{T}_1,\mathcal{T}_2\}$, by applying \eqref{eq:sigma}, Assumption \ref{asp:smth}, and \eqref{ieq:wk_moveout}.
By inserting \eqref{ieq:qntUlDl} in \eqref{ieq:c_part1} and recovering $\tau_i$ by $t-i$,
\begin{align}
    &\mathbb{E}\Big[\delta_t\Big|\Big|\boldsymbol{G}(t\boldsymbol{1})-\boldsymbol{G}\left(\{t-a_k^{(t)}\}_{\forall k}\right)\Big|\Big|\Big]\leq L^2\hat{\sigma}\eta\nonumber\\
	&\quad\quad \cdot \mathbb{E}\Bigg\{\sum_{k\in\mathcal{K}}w_k\sum_{i=1}^{a_k^{(t)}}\delta_t\left[\sum_{j=r_{t-i+1}}^{t-i}\sum_{l\in\mathcal{K}}w_l\|\tilde{\boldsymbol{\theta}}(j-a_l^{(j)})-\boldsymbol{\theta}_l^*\|\right]\Bigg\}\nonumber\\
	&\leq L^2\hat{\sigma}a_{\lim}\eta\mathbb{E}\Bigg\{\delta_t\left[\sum_{j=r_{t-i+1}}^{t-1}\sum_{l\in\mathcal{K}}w_l\|\tilde{\boldsymbol{\theta}}(j-a_l^{(j)})-\boldsymbol{\theta}_l^*\|\right]\Bigg\}\label{ieq:usingAlim}\\
    &\leq L^2\hat{\sigma}a_{\lim}\eta\mathbb{E}\Bigg\{\delta_t\left[\sum_{j=r_{t-i+1}}^{t-1}\left(\sum_{l\in\mathcal{K}}w_l\delta_{j-a_l^{(j)}}+\sqrt{\zeta}\right)\right]\Bigg\},\label{ieq:usingTriAng}
\end{align}
where \eqref{ieq:usingAlim} is based on $a_k^{(t)}\leq a_{\lim}$ in Assumption \ref{asp:ageLim}; \eqref{ieq:usingTriAng} is by applying the triangle inequality and $$\sum\nolimits_{k\in\mathcal{K}}w_k\|\boldsymbol{\theta}^*-\boldsymbol{\theta}_k^*\|\leq\sqrt{\sum\nolimits_{k\in\mathcal{K}}w_k\|\boldsymbol{\theta}^*-\boldsymbol{\theta}_k^*\|^2}=\sqrt{\zeta}.$$
Since $(t-a_{\lim}+1)-r_{t-a_{\lim}+1}\leq a_{\lim}$, 
\begin{equation*}
\sum\nolimits_{j=r_{t-i+1}}^{t-1}1\leq t-r_{t-a_{\lim}+1}\leq 2a_{\lim}~\text{and }r_{t-a_{\lim}+1}\geq t-2a_{\lim}.
\end{equation*}
\eqref{ieq:usingTriAng} can then be rearranged as
\begin{align*}
	&\mathbb{E}\Bigg[\delta_t\Big|\Big|\boldsymbol{G}(t\boldsymbol{1})-\boldsymbol{G}\left(\{t-a_k^{(t)}\}_{\forall k}\right)\Big|\Big|\Bigg]\nonumber\\
	&\leq L^2\hat{\sigma}a_{\lim}\eta\left[2a_{\lim}\sqrt{\zeta}\mathbb{E}\left[\delta_t\right]+\mathbb{E}\left[\delta_t\sum\nolimits_{i=1}^{\min\left(t-1,3a_{\lim}\right)}\delta_{t-i}\right]\right],\label{ieq:lm_preFinal}
\end{align*}
The proof is complete by relaxing the bound above with: 
\begin{itemize}
	\item $\delta_t<\delta_t^2+1$, since $x<x^2+1, \forall x$
	\item the arithmetic-geometric mean inequality, leading to
	\begin{align*}
		&\delta_t\sum\nolimits_{i=1}^{\min\left(t-1,3a_{\lim}\right)}\delta_{t-i}\leq\sum\nolimits_{i=1}^{\min\left(t-1,3a_{\lim}\right)}\left(\delta_{t}^2+\delta_{t-i}^2\right)/2\nonumber\\
		&\leq 3a_{\lim}\delta_{t}^2/2+\sum\nolimits_{i=1}^{\min\left(t-1,3a_{\lim}\right)}\delta_{t-i}^2/2.
	\end{align*}
\end{itemize}

\bibliographystyle{IEEEtran}
\bibliography{ref}

	\begin{IEEEbiography}[{\includegraphics[width=1in,height=1.25in,clip,keepaspectratio]{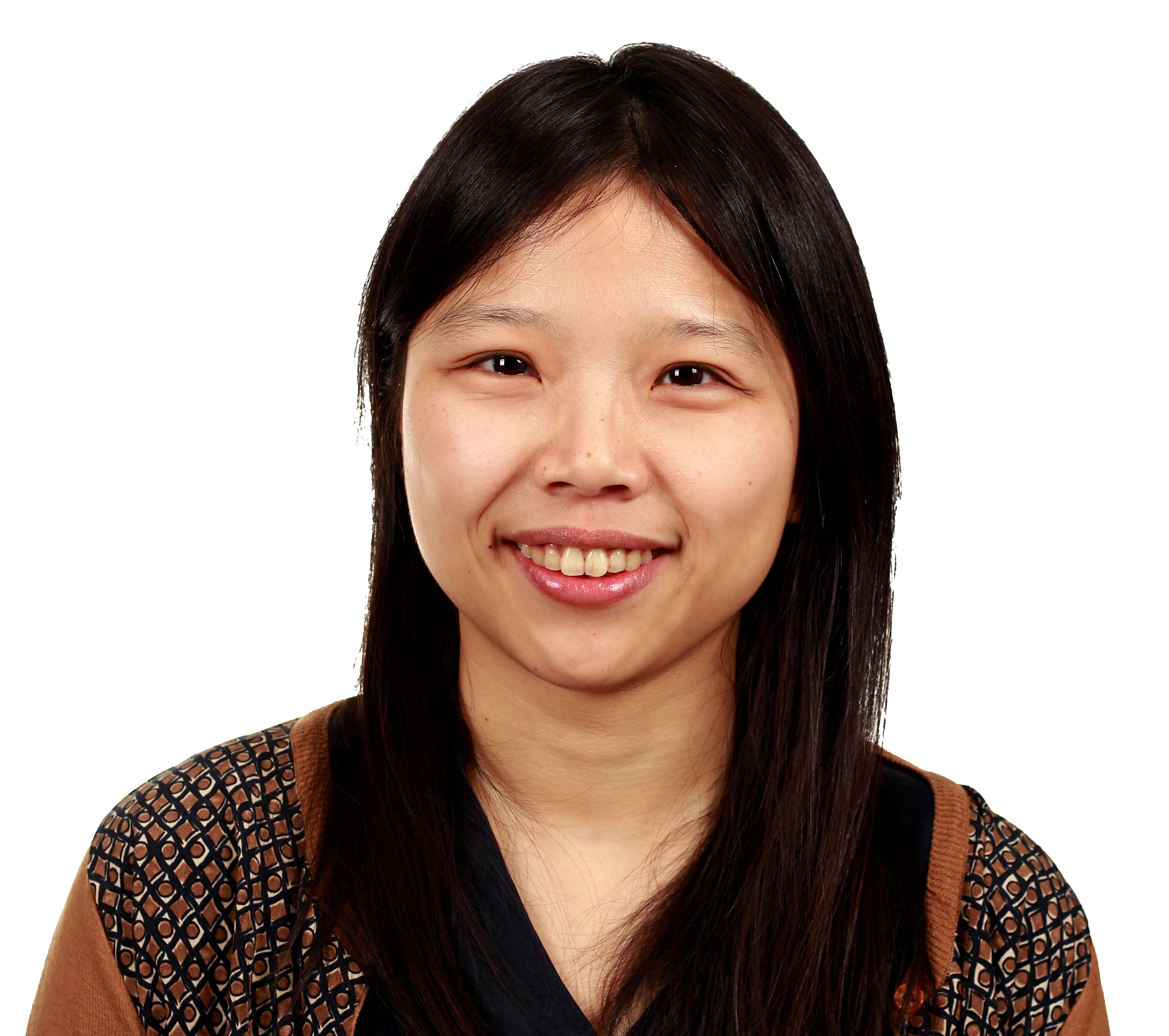}}]
		{\bfseries Chung-Hsuan Hu} received the B.Sc. degree in electronics and electrical engineering, and M.Sc. degree in communications engineering from National Yang Ming Chiao Tung University (NYCU), Taiwan, in 2010 and 2012, respectively. From 2013 to 2020, she worked as a communication systems engineer with MediaTek Inc., Taiwan. After that, in 2026, she received the Ph.D. degree with the Division of Communication Systems, Department of Electrical Engineering, Linköping University, Sweden. Her research interests include distributed learning systems and wireless communications.
	\end{IEEEbiography}
	\begin{IEEEbiography}[{\includegraphics[width=1in,height=1.25in,clip,keepaspectratio]{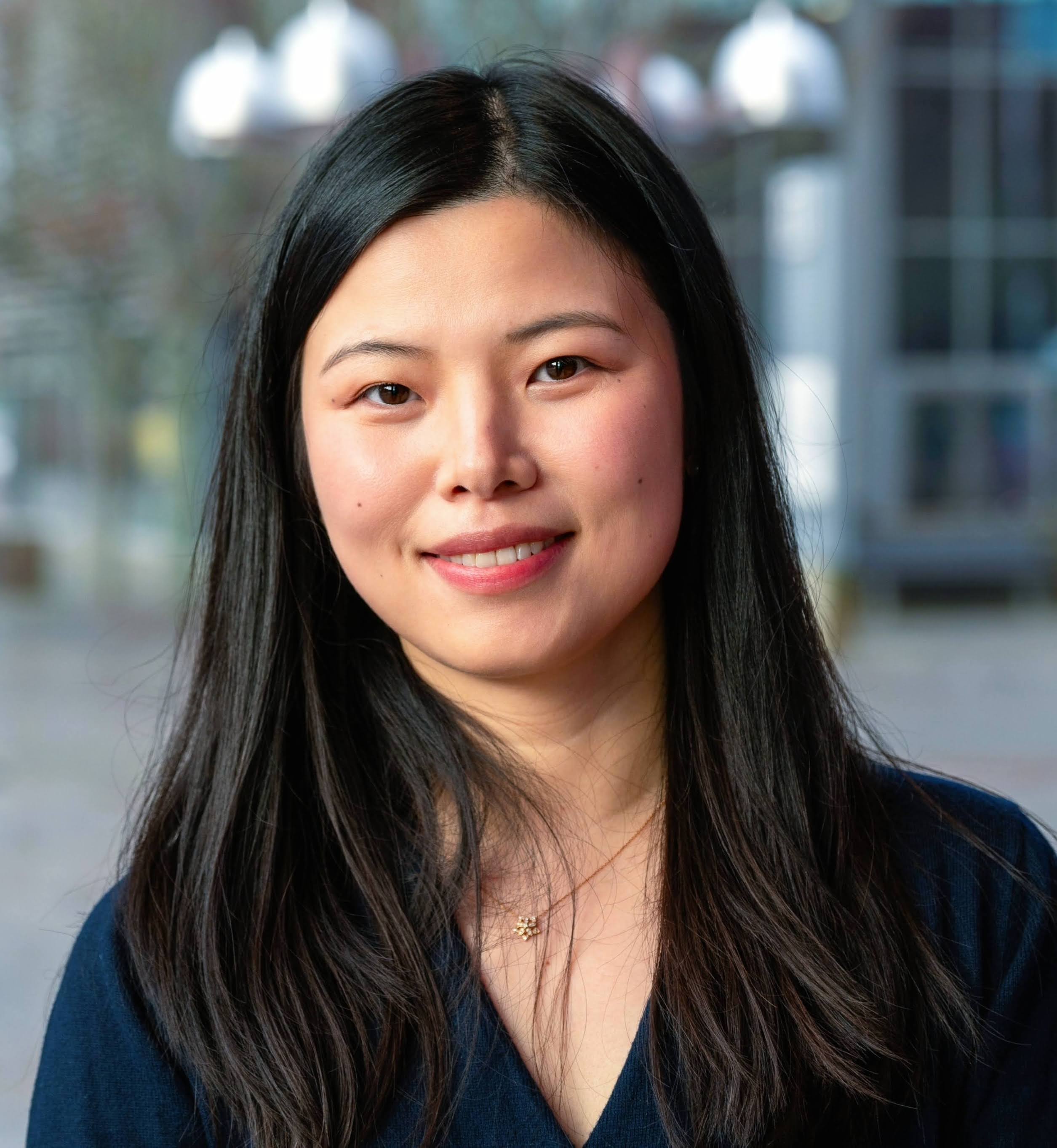}}]
{\bfseries Zheng Chen} is an Associate Professor with the Department of Electrical Engineering at Linköping University, Sweden. She received her M.Sc. and Ph.D. degrees from CentraleSupélec, Université Paris-Saclay, France, in 2013 and 2017, respectively. Her research focuses on distributed and cooperative computing, optimization, and learning for networked intelligent systems under communication constraints.
She was a recipient of the 2020 IEEE Communications Society Young Author Best Paper Award. She has served as the co-chair of several workshops and special sessions at IEEE GLOBECOM, ICASSP, SPAWC, Asilomar, and as the technical program chair of the 2022 IEEE SPS-EURASIP summer school on “Defining 6G: Theory, Applications and Enabling Technologies”. She is currently an Associate Editor of the IEEE Transactions on Wireless Communications, IEEE Transactions on Communications, and IEEE Transactions on Green Communications and Networking. She is also serving as the Lead Guest Editor for IEEE JSAC special issue on “Distributed Optimization, Learning, and Inference over Communication-Constrained Networks”.
	\end{IEEEbiography}
	\begin{IEEEbiography}[{\includegraphics[width=1in,height=1.25in,clip,keepaspectratio]{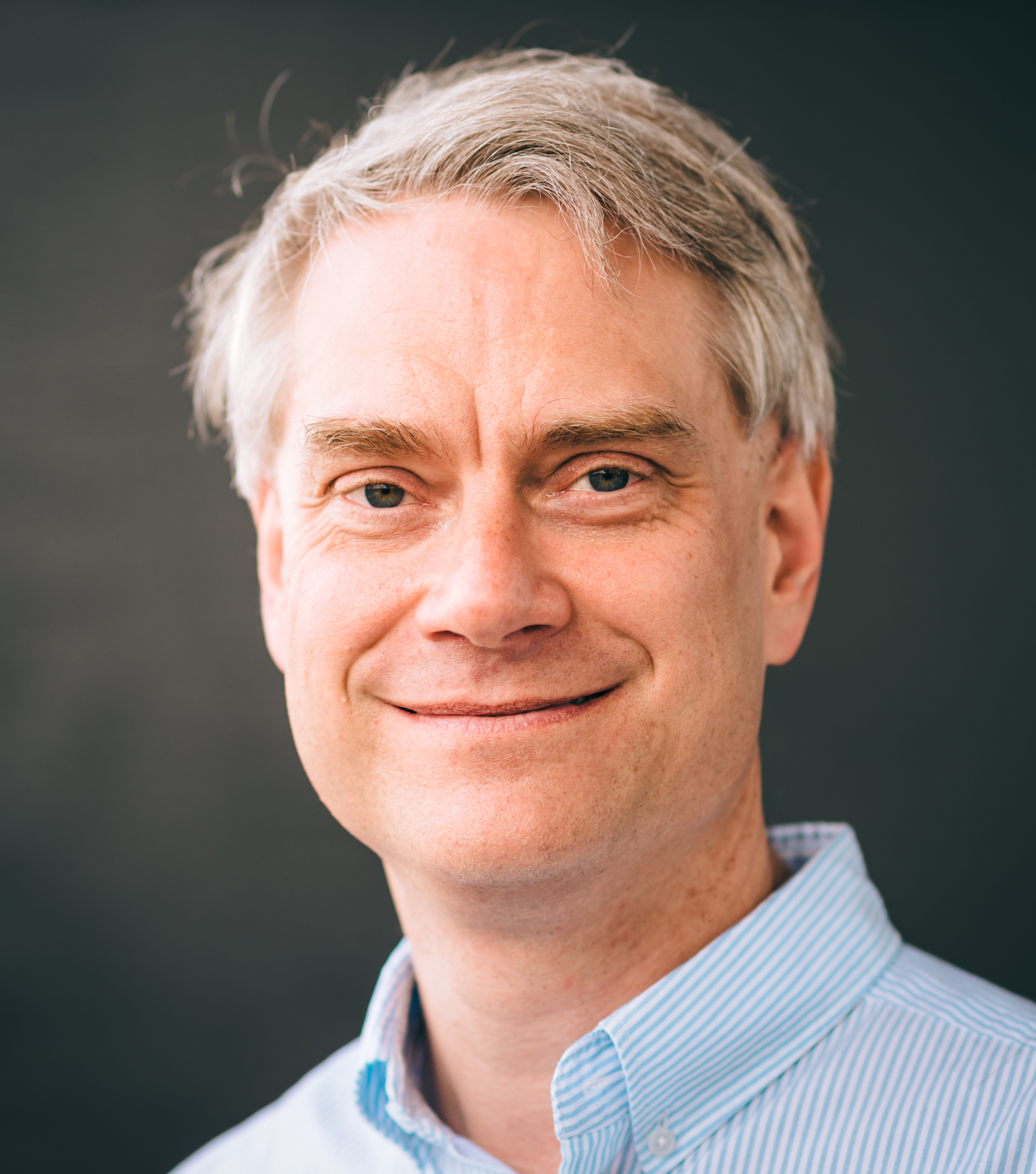}}]
		{\bfseries Erik G. Larsson} (Fellow) received the Ph.D. degree from Uppsala University,
		Uppsala, Sweden, in 2002.  He is currently Professor of Communication
		Systems at Link\"oping University (LiU) in Link\"oping, Sweden. He was
		with the KTH Royal Institute of Technology in Stockholm, Sweden, the
		George Washington University, USA, the University of Florida, USA, and
		Ericsson Research, Sweden.  His main professional interests are within
		wireless communications, signal processing, and network science. He 
		co-authored \emph{Space-Time Block Coding for  Wireless Communications} (Cambridge University Press, 2003) 
		and \emph{Fundamentals of Massive MIMO} (Cambridge University Press, 2016). 	
		
		He served as  chair  of the IEEE Signal Processing Society SPCOM technical committee (2015--2016), 
		chair of  the \emph{IEEE Wireless  Communications Letters} steering committee (2014--2015), 
		member of the  \emph{IEEE Transactions on Wireless Communications}    steering committee (2019-2022),
		General and Technical Chair of the Asilomar SSC conference (2015, 2012), 
		technical co-chair of the IEEE Communication Theory Workshop (2019), 
		and   member of the  IEEE Signal Processing Society Awards Board (2017--2019).
		He was Associate Editor for, among others, the \emph{IEEE Transactions on Communications} (2010-2014), 
		the \emph{IEEE Transactions on Signal Processing} (2006-2010),
		and  the \emph{IEEE Signal  Processing Magazine} (2018-2022).
		
		He received the IEEE Signal Processing Magazine Best Column Award
		twice, in 2012 and 2014, the IEEE ComSoc Stephen O. Rice Prize in
		Communications Theory in 2015, the IEEE ComSoc Leonard G. Abraham
		Prize in 2017, the IEEE ComSoc Best Tutorial Paper Award in 2018, 
		the IEEE ComSoc Fred W. Ellersick Prize in 2019, and the
		2023 IEEE SPS Donald G. Fink Overview Paper Award,
and the IEEE ComSoc Test of Time Paper Award for Advances in Communications in 2026.
	\end{IEEEbiography}
    
\end{document}